\documentclass{LMCS}

\def\doi{7 (3:19) 2011}
\lmcsheading%
{\doi}
{1--42}
{}
{}
{Nov.~\phantom04, 2010}
{Sep.~30, 2011}
{}   

\usepackage{amsmath}
\usepackage{amssymb}
\usepackage{latexsym}
\usepackage{rotating}
\usepackage{times}
\usepackage{amsthm}
\usepackage{algorithm}
\usepackage{algpseudocode}
\usepackage{url}
\usepackage{pgf}
\usepackage{lscape}
\usepackage{enumerate,hyperref}

\usepackage{ifpdf}
\ifpdf
\else
\fi

\begin{document}

\theoremstyle{plain}
\newtheorem{theorem}{Theorem}[section]
\newtheorem{lemma}[theorem]{Lemma}
\newtheorem{corollary}[theorem]{Corollary}

\newcommand{\pgsolver}{{\sc PGSolver}\,}

\newcommand{\rew}{\mathtt{rew}}
\newcommand{\ran}{\mathtt{ran}}
\newcommand{\NN}{\mathbb{N}}
\newcommand{\ZZ}{\mathbb{Z}}
\newcommand{\RR}{\mathbb{R}}
\newcommand{\boldzero}{\mathbf{0}}
\newcommand{\boldone}{\mathbf{1}}

\newcommand{\NP}{\mathtt{NP}}
\newcommand{\PP}{\mathtt{P}}
\newcommand{\LP}{\mathtt{LP}}
\newcommand{\coNP}{\mathtt{coNP}}
\newcommand{\UP}{\mathtt{UP}}
\newcommand{\CF}{\mathtt{CF}}
\newcommand{\coUP}{\mathtt{coUP}}
\newcommand{\BigO}[1]{\mathcal{O}({#1})}
\newcommand{\strategyset}[2]{\mathcal{S}_{#2}({#1})}

\newcommand{\inducedmpg}[1]{{\mathit{IndMPG}({#1})}}
\newcommand{\induceddpg}[1]{{\mathit{IndDPG}({#1})}}
\newcommand{\inducedssg}[1]{{\mathit{IndSSG}({#1})}}

\newcommand{\dlroot}{\mathit{root}}
\newcommand{\dlindex}{\mathit{ind}}

\newcommand{\declaneeven}{a}
\newcommand{\declaneodd}{t}
\newcommand{\declaneroot}{c}
\newcommand{\cyclenode}{d}
\newcommand{\cyclecenter}{e}
\newcommand{\cycleaccess}{f}
\newcommand{\cycleselector}{g}
\newcommand{\cycleleaver}{h}
\newcommand{\upperselector}{k}
\newcommand{\finalcycle}{x}
\newcommand{\bitselector}{r}
\newcommand{\starteven}{s}
\newcommand{\cycleselectorx}{y}
\newcommand{\cyclenodex}{\cyclenode^1}
\newcommand{\cyclenodey}{\cyclenode^2}
\newcommand{\cyclenodez}{\cyclenode^3}

\newcommand{\noindentpara}[1]{{\parindent0pt {#1}}}

\title{An Exponential Lower Bound for the Latest Deterministic Strategy Iteration Algorithms}

\author[O.~Friedmann]{Oliver Friedmann}
\address{University of Munich}
\email{Oliver.Friedmann@gmail.com}

\begin{abstract}
This paper presents a new exponential lower bound for the two most popular deterministic variants of the strategy improvement algorithms for solving parity, mean payoff, discounted payoff and simple stochastic games. The first variant improves every node in each step maximizing the current valuation locally, whereas the second variant computes the globally optimal improvement in each step. We outline families of games on which both variants require exponentially many strategy iterations.
\end{abstract}


\maketitle

\section{Introduction} \label{section: introduction}
In this paper, we study lower bounds for strategy improvement algorithms for solving
parity games, mean payoff games, discounted payoff games as well as simple stochastic
games. These are two-player games of perfect information played on directed graphs,
and are related by a chain of polynomial-time reductions.

Parity games can be reduced to mean payoff games \cite{puri/phd}, mean payoff games to
discounted payoff games, and the latter ones to simple stochastic games \cite{zwickpaterson/1996}.
Solving games of any of these classes is one of the few combinatorial problems that belongs to the
complexity class $\NP \cap \coNP$ and that is not (yet) known to belong to $\PP$
\cite{Emerson93a,condon92thecomplexity}. It has also been shown that solving parity games as well
as mean and discounted payoff games belongs to $\UP \cap \coUP$ \cite{Jurdzinski/98}.

We mainly consider parity games in this paper. They are played on a directed graph that is
partitioned into two node sets associated with the two players; the nodes are labeled with
natural numbers, called priorities. A play in a parity game is an infinite sequence of nodes
whose winner is determined by the parity of the highest priority that occurs infinitely
often, giving parity games their name.

The reason why parity games seem to be the most appropriate class of games,
when trying to construct a worst-case family for one of the four classes, is that the effect of each node in a parity
game is very clear: a higher priority dominates all lower priorities (in a play), no matter
how many there are. By showing that the strategy iteration on our family of parity games 
directly corresponds to the strategy iteration that solves the other classes of games, we
get the lower bounds for these by applying the standard reductions to our games.

Parity games occur in several fields of theoretical computer science, e.g.\ as solution
to the problem of emptiness of tree automata \cite{lncs2500,focs91*368}
or as algorithmic backend to the model checking problem of the modal $\mu$-calculus \cite{Emerson93a,Stirling95}.

There are many algorithms that solve parity games, such as the recursive decomposing algorithm due
to Zielonka \cite{TCS::Zielonka1998} and its recent improvement by Jurdzi{\'n}ski, Paterson and Zwick
\cite{JPZ06}, the small progress measures algorithm due to Jurdzi{\'n}ski \cite{Jurdzinski/00} with
its recent improvement by Schewe \cite{Schewe/07/Parity}, the model-checking algorithm due to Stevens
and Stirling \cite{StevensStirling98} and finally the two strategy improvement algorithms by V\"oge and
Jurdzi{\'n}ski \cite{conf/cav/VogeJ00} and Schewe \cite{conf/csl/Schewe08}.

All mentioned algorithms except for the two deterministic subexponential algorithms
\cite{JPZ06,Schewe/07/Parity} and except for the two strategy improvement algorithms have been shown to have
a superpolynomial or exponential worst-case runtime complexity at best
\cite{Jurdzinski/00,FriedmannIJFCS2010,FriedmannRecursive2010}. The currently best known
upper bound on the deterministic solution of parity games is
$\BigO{|E| \cdot |V|^{\frac{1}{3}|\ran \Omega|}}$ due to Schewe's big-step algorithm \cite{Schewe/07/Parity}.

The \emph{strategy improvement}, \emph{strategy iteration} or \emph{policy iteration} technique is the most general
approach that can be applied as a solving procedure for all of these game classes. It was introduced by Howard
\cite{howard/1960} for solving problems on Markov decision processes and has been adapted by several other authors
for solving nonterminating stochastic games \cite{hoffmankarp/1966}, simple stochastic games
\cite{condon92thecomplexity}, discounted and mean payoff games \cite{puri/phd,zwickpaterson/1996} as well as
parity games \cite{conf/cav/VogeJ00}.

Strategy iteration is an algorithmic scheme that is parameterized by an \emph{improvement policy} which basically
defines how to select a successor strategy in the iteration process. There are two major kinds of improvement
policies: deterministic and randomized approaches; we will investigate deterministic approaches in this paper.

For discounted payoff games, there is the deterministic algorithm due to Puri \cite{puri/phd} that can also be
used to solve mean payoff games as well as parity games by reduction \cite{zwickpaterson/1996,conf/cav/VogeJ00}.
V\"oge and Jurdzi{\'n}ski's improvement algorithm is a refined version of Puri's on parity games that omits the use of
high-precision rational numbers; there are at least two reasonable improvement policies for the V\"oge-Jurdzi{\'n}ski procedure
appearing in the literature such as the standard \emph{locally optimizing policy} and Schewe's
\emph{globally optimizing policy}.

An example has been known for some time for which a sufficiently poor choice of a single-switch policy causes an
exponential number of iterations of the strategy improvement algorithm \cite{BjoerklundVorobyov/2007},
but there have been no games known so far on which the policies due to V\"oge/Jurdzi{\'n}ski or Schewe require more
than linearly many iterations.

In this paper, we particularly investigate the locally optimizing policy -- which is, by far, the most natural choice
for a multi-switching improvement policy -- for solving parity games as it is applied
by default in the original paper of V\"oge and Jurdzi{\'n}ski. We present a family of games comprising a linear number of nodes and a
quadratic number of edges such that the strategy improvement algorithm using this policy requires an
exponential number of iterations on them. We explain how these games can be refined in such a way that
they only comprise a linear number of edges resulting in an undeniable exponential lower bound. Additionally,
we describe what parts of the games have to be altered in order to get a family that results in exponentially
many iterations when solved by Schewe's strategy improvement algorithm.

Finally, we show that the parity game strategy iteration on our games directly corresponds to the strategy
iteration that solves the associated mean payoff, discounted payoff as well as simple stochastic games,
resulting in an exponential lower bound for the standard strategy improvement algorithms for all of these
game classes.

Section~\ref{section: parity games} defines the basic notions of parity games and some notations that are employed throughout the paper. Section~\ref{section: strategy improvement} recaps the strategy improvement algorithm by V\"oge and Jurdzi{\'n}ski; we define the two considered improvement policies in Section~\ref{section: improvement policies}. In Section~\ref{section: sink games}, we define a subclass of parity games called sink games that allows us to relate the lower bounds for parity games to the other games classes. Section~\ref{section: locally pol} presents a family of games on which the locally improving algorithm requires an exponential number of iterations. We discuss some improvements of the family in Section~\ref{section: improvements}. In Section~\ref{section: globally pol}, we consider the modifications that have to be applied to our construction to obtain a lower bound for the globally optimizing policy. In Section~\ref{section: othergames}, we show how to transfer the lower bounds to mean payoff, discounted payoff and simple stochastic games.

\section{Parity Games} \label{section: parity games}
A \emph{parity game} is a tuple $G = (V,V_0,V_1,E,\Omega)$ where $(V,E)$ forms a directed graph whose node set is partitioned into $V = V_0 \cup V_1$ with $V_0 \cap V_1 = \emptyset$, and $\Omega : V \to \NN$ is the \emph{priority function} that assigns to each node a natural number called the \emph{priority} of the node. We assume the graph to be total, i.e.\ for every $v \in V$ there is a $w \in V$ s.t.\ $(v,w) \in E$.

In the following we will restrict ourselves to finite parity games. 
W.l.o.g.\ we assume $\Omega$ to be injective, i.e.\ there are no two different nodes with the same priority.

We also use infix notation $vEw$ instead of $(v,w) \in E$ and define the set of all \emph{successors} of $v$ as $vE := \{ w \mid vEw \}$. The size $|G|$ of a parity game $G = (V,\ V_0,\ V_1,\ E,\ \Omega)$ is defined to be the cardinality of $E$, i.e.\ $|G| := |E|$; since we assume parity games to be total w.r.t.\ $E$, this is a reasonable way to measure the size.

The game is played between two players called $0$ and $1$: starting in a node $v_0 \in V$, they construct an infinite path through the graph as follows. If the construction so far has yielded a finite sequence $v_0\ldots v_n$ and $v_n \in V_i$ then player $i$ selects a $w \in v_nE$ and the play continues with $v_0\ldots v_n w$.

Every play has a unique winner given by the \emph{parity} of the greatest priority that occurs infinitely often. The winner of the play $v_0 v_1 v_2 \ldots$ is player $i$ iff $\max \{ p \mid \forall j \in \NN\, \exists k \geq j:\, \Omega(v_k) = p \} \equiv_2 i$  (where $i \equiv_k j$ holds iff $|i - j| \mod k = 0$). That is, player 0 tries to make an even priority occur infinitely often without any greater odd priorities occurring infinitely often, player 1 attempts the converse.

We depict parity games as directed graphs where nodes owned by player 0 are drawn as circles and nodes owned by player 1 are drawn as rectangles; all nodes are labeled with their respective priority, and -- if needed -- with their name.

A \emph{strategy} for player $i$ is a -- possibly partial -- function $\sigma: V^*V_i \to V$, s.t.\ for all sequences $v_0 \ldots v_n$ with $v_{j+1} \in v_jE$ for all $j=0,\ldots,n-1$, and all $v_n \in V_i$ we have: $\sigma(v_0\ldots v_n) \in v_nE$. A play $v_0 v_1 \ldots$ \emph{conforms} to a strategy $\sigma$ for player $i$ if for all $j \in \NN$ we have: if $v_j \in V_i$ then $v_{j+1} = \sigma(v_0\ldots v_j)$. Intuitively, conforming to a strategy means to always make those choices that are prescribed by the strategy. A strategy $\sigma$ for player $i$ is a \emph{winning strategy} in node $v$ if player $i$ wins every play that begins in $v$ and conforms to $\sigma$.

A strategy $\sigma$ for player $i$ is called \emph{positional} if for all $v_0\ldots v_n \in V^*V_i$ and all $w_0\ldots w_m \in V^*V_i$ we have: if $v_n = w_m$ then $\sigma(v_0\ldots v_n) = \sigma(w_0\ldots w_m)$. That is, the choice of the strategy on a finite path only depends on the last node on that path.

With $G$ we associate two sets $W_0,W_1 \subseteq V$; $W_i$ is the set of all nodes $v$ s.t.\ player
$i$ wins the game $G$ starting in $v$. Here we restrict ourselves to positional strategies because it is
well-known that a player has a (general) winning strategy iff she has a positional winning strategy for a
given game. In fact, parity games enjoy positional determinacy meaning that for every node
$v$ in the game either $v \in W_0$ or $v \in W_1$ \cite{focs91*368}. Furthermore,
it is not difficult to show that, whenever player $i$ has winning strategies $\sigma_v$ for all $v \in U$ for
some $U \subseteq V$, then there is also a single strategy $\sigma$ that is winning for player $i$ from every
node in $U$.

The problem of solving a parity game is to compute $W_0$ and $W_1$ as well as corresponding winning
strategies $\sigma_0$ and $\sigma_1$ for the players on their respective winning regions.

A strategy $\sigma$ for player $i$ induces a \emph{strategy subgame} $G|_\sigma := (V, V_0, V_1, E|_\sigma, \Omega)$ where $E|_\sigma := \{(u, v) \in E \mid u \in dom(\sigma) \Rightarrow \sigma(u) = v\}$. Such a subgame $G|_\sigma$ is basically the same game as $G$ with the restriction that whenever $\sigma$ provides a strategy decision for a node $u \in V_i$, all transitions from $u$ but $\sigma(u)$ are no longer accessible. The set of strategies for player $i$ is denoted by $\strategyset{G}{i}$.

\section{Strategy Improvement} \label{section: strategy improvement}
We briefly recap the basic definitions of the strategy improvement algorithm.
For a given parity game $G = (V,\ V_0,\ V_1,\ E,\ \Omega)$, the \emph{reward} of node $v$ is defined as follows: $\rew_G(v) := \Omega(v)$ if $\Omega(v) \equiv_2 0$ and $\rew_G(v) := -\Omega(v)$ otherwise. The set of \emph{even resp.\ odd priority nodes} is defined to be $V_\oplus := \{v \in V \mid \Omega(v) \equiv_2 0\}$ resp.\ $V_\ominus := \{v \in V \mid \Omega(v) \equiv_2 1\}$.

The \emph{relevance ordering} $<$ on $V$ is induced by $\Omega$: $v < u :\iff \Omega(v) < \Omega(u)$; additionally one defines the \emph{reward ordering} $\prec$ on $V$ by $v \prec u :\iff \rew_G(v) < \rew_G(u)$. Note that both orderings are total due to injectivity of the priority function.

Let $\pi$ be a path $\sigma$ be a strategy for player~$i$. We say that $\pi$ \emph{conforms} to $\sigma$ iff for every $j$ with $\pi(j) \in V_i$ we have $\sigma(\pi(j)) = \pi(j+1)$.

Let $v$ be a node, $\sigma$ be a positional player~0 strategy and $\tau$ be a positional player~1 strategy. Starting in $v$, there is exactly one path $\pi_{\sigma,\tau,v}$ that conforms to $\sigma$ and $\tau$. Since $\sigma$ and $\tau$ are positional strategies, this path can be uniquely written as follows.
\begin{displaymath}
\pi_{\sigma,\tau,v} = v_1\ldots v_k(w_1\ldots w_l)^\omega
\end{displaymath}
with $v_1 = v$, $v_i \not= w_1$ for all $1 \leq i \leq k$ and $\Omega(w_1) > \Omega(w_j)$ for all $1 < j \leq l$. Note that the uniqueness follows from the fact that all nodes on the cycle have different priorities and we choose $w_1$ to be the node with highest priority.

Discrete strategy improvement relies on a more abstract description of such a play $\pi_{\sigma,\tau,v}$. In fact, we only consider the \emph{dominating cycle node} $w_1$, the set of \emph{more relevant nodes} -- i.e.\ all $v_i > w_1$ -- on the path to the cycle node, and the \emph{length} $k$ of the path leading to the cycle node. More formally, the \emph{node valuation of $v$ w.r.t.\ $\sigma$ and $\tau$} is defined as follows.
\begin{displaymath}
\vartheta_{\sigma,\tau,v} := (w_1, \{v_i > w_1 \mid 1 \leq i \leq k\}, k)
\end{displaymath}
Given a node valuation $\vartheta$, we refer to $w_1$ as the \emph{cycle component}, to $\{v_i > w_1 \mid 1 \leq i \leq k\}$ as the \emph{path component}, and to $k$ as the \emph{length component} of $\vartheta$.

In order to compare node valuations with each other, we introduce a total ordering on the set of node valuations. For that reason, we need to define a total ordering $\prec$ on the second component of node valuations -- i.e.\ on subsets of $V$ -- first. To compare two different sets $M$ and $N$ of nodes, we order
all nodes lexicographically w.r.t.\ to their relevance and consider the first position in which the two lexicographically ordered sets differ, i.e.\ there is
a node $v \in M$ and a node $w \in N$ with $v \not= w$ s.t.\ $u \in M$ iff $u \in N$ for all $u > v$ and all $u > w$. Now $N$ is better than $M$ iff $v \prec w$, i.e.\ the set which gives the higher reward in the first differing position is superior to the other set.

In other words, to determine which set of nodes is better w.r.t.\ $\prec$, one considers the node with the highest priority that occurs in only one of the two sets. The set owning that node is greater than the other if and only if that node has an even priority. More formally:
\begin{displaymath}
M \prec N  \, : \iff \, M \triangle N \not= \emptyset \textrm{ and } \max_<(M \triangle N) \in ((N \cap V_\oplus) \cup (M \cap V_\ominus))
\end{displaymath}
where $M \triangle N$ denotes the symmetric difference of both sets.

Now we are able to extend the total ordering on sets of nodes to node valuations
The motivation behind this ordering is a lexicographic measurement of the profitability of a positional play w.r.t.\ player~0: the most prominent part of a positional play is the cycle in which the plays eventually stays, and here it is the reward ordering on the dominating cycle node that defines the profitability for player~0. The second important part is the loopless path that leads to the dominating cycle node. Here, we measure the profitability of a loopless path by a \emph{lexicographic} ordering on the \emph{relevancy} of the nodes on path, applying the \emph{reward} ordering on each component in the lexicographic ordering. Finally, we consider the length, and the intuition behind the definition is that, assuming we have an even-priority dominating cycle node, it is better to reach the cycle fast whereas it is better to stay as long as possible out of the cycle otherwise.
More formally:
\begin{displaymath}
(u, M, e) \prec (v, N, f) \, : \iff \,
\begin{cases}
\left(u \prec v\right) \textrm{ or } \left(u = v \textrm{ and } M \prec N \right) \textrm{ or } \\
\left(u = v \textrm{ and } M = N \textrm{ and } e < f \textrm{ and } u \in V_\ominus \right) \textrm{ or } \\
\left(u = v \textrm{ and } M = N \textrm{ and } e > f \textrm{ and } u \in V_\oplus \right)
\end{cases}
\end{displaymath}

Given a player~0 strategy $\sigma$, it is our goal to find a best response counterstrategy $\tau$ that minimizes the associated node valuations.
A strategy $\tau$ is an \emph{optimal counterstrategy} w.r.t.\ $\sigma$ iff for every opponent strategy $\tau'$ and for every node $v$ we have: $\vartheta_{\sigma,\tau,v} \preceq \vartheta_{\sigma,\tau',v}$.

It is well-known that an optimal counterstrategy always exists and that it is efficiently computable.
\begin{lemma}[\cite{conf/cav/VogeJ00}]
Let $G$ be a parity game and $\sigma$ be a player~0 strategy. An optimal counterstrategy for player~1 w.r.t.\ $\sigma$ exists and can be computed in polynomial time.
\end{lemma}

A fixed but arbitrary optimal counterstrategy will be denoted by $\tau_\sigma$ from now on. The associated \emph{game valuation} $\Xi_\sigma$ is a map that assigns to each node the node valuation w.r.t.\ $\sigma$ and $\tau_\sigma$:
\begin{displaymath}
\Xi_\sigma: v \mapsto \vartheta_{\sigma,\tau_\sigma,v}
\end{displaymath}

Game valuations are used to measure the performance of a strategy of player~0: for a fixed strategy $\sigma$ of player~0 and a node $v$, the associated valuation essentially states which is the worst cycle that can be reached from $v$ conforming to $\sigma$ as well as the worst loopless path leading to that cycle (also conforming to $\sigma$).

We also write $v \prec_\sigma u$ to compare the $\Xi_\sigma$-valuations of two nodes, i.e.\ to abbreviate $\Xi_\sigma(v) \prec \Xi_\sigma(u)$.

A run of the strategy improvement algorithm can be expressed by a sequence of \emph{improving} game valuations; a partial ordering on game valuations is quite naturally defined as follows:
\begin{displaymath}
\Xi \lhd \Xi' \, : \iff \, \left(\Xi(v) \preceq \Xi'(v) \textrm{ for all } v \in V\right) \textrm{ and } \left(\Xi \not= \Xi'\right)
\end{displaymath}

A valuation $\Xi_\sigma$ can be used to create a new strategy of player~0. The strategy improvement algorithm is only allowed to select new strategy decisions for player~0 occurring in the \emph{improvement arena} $\mathcal{A}_{G,\sigma} := (V,\ V_0,\ V_1,\ E',\ \Omega)$ where
\begin{eqnarray*}
&vE'u \, : \iff&\\
&vEu \textrm{ and } \left(v \in V_1 \textrm{ or } (v \in V_0 \textrm{ and } \sigma(v) \preceq_\sigma u)\right)&
\end{eqnarray*}

Thus all edges performing worse than the current strategy are removed from the game. A strategy $\sigma$ is \emph{improvable} iff there is a node $v \in V_0$, a node $u \in V$ with $vEu$ s.t.\ $\sigma(v) \prec_\sigma u$.

An \emph{improvement policy} now selects a strategy for player~0 in a given improvement arena. More formally: an improvement policy is a map $\mathcal{I}_G: \strategyset{G}{0} \rightarrow \strategyset{G}{0}$ fulfilling the following two conditions for every strategy $\sigma$.
\begin{enumerate}[(1)]
\item For every node $v \in V_0$ it holds that $(v,\mathcal{I}_G(\sigma)(v))$ is an edge in $\mathcal{A}_{G,\sigma}$.
\item If $\sigma$ is improvable then there is a node $v \in V_0$ s.t.\ $\sigma(v) \prec_\sigma \mathcal{I}_G(\sigma)(v)$.
\end{enumerate}
We say that an edge $(v,u)$ is an \emph{improving edge} w.r.t.\ $\sigma$ iff $v \in V_0$, $u \in vE$, $\sigma(v) \not= u$ and $\sigma(v) \prec_\sigma u$.

Jurdzi{\'n}ski and V\"oge proved in their work that every strategy that is improved by an improvement policy can only result in strategies with valuations strictly better than the valuation of the original strategy.

\begin{theorem}[\cite{conf/cav/VogeJ00}]
Let $G$ be a parity game, $\sigma$ be an improvable strategy and $\mathcal{I}_G$ be an improvement policy. We have $\Xi_\sigma \lhd \Xi_{\mathcal{I}_G(\sigma)}$.
\end{theorem}

If a strategy is not improvable, the strategy iteration procedure comes to an end and the winning sets for both players as well as associated winning strategies can be easily derived from the given valuation.

\begin{theorem}[\cite{conf/cav/VogeJ00}]\label{theorem: discrete strategy improvement winning theorem}
Let $G$ be a parity game and $\sigma$ be a non-improvable strategy. Then the following holds:
\begin{enumerate}[\em(1)]
\item $W_0 = \{v \mid \Xi_\sigma(v) = (w, \_, \_) \textrm{ and } w \in V_\oplus\}$
\item $W_1 = \{v \mid \Xi_\sigma(v) = (w, \_, \_) \textrm{ and } w \in V_\ominus\}$
\item $\sigma$ is a winning strategy for player 0 on $W_0$
\item $\tau_\sigma$ is a winning strategy for player 1 on $W_1$
\item $\sigma$ is $\unlhd$-optimal
\end{enumerate}\smallskip
\end{theorem}

\noindent The strategy iteration starts with an initial strategy $\iota_G$ and runs for a given
improvement policy $\mathcal{I}_G$ as outlined in the pseudo-code of
Algorithm~\ref{algorithm: discrete strategy improvement}.

\begin{algorithm}[!h]
\begin{algorithmic}[1]
\State $\sigma \gets \iota_G$
\While {$\sigma$ is improvable}
	\State $\sigma \gets \mathcal{I}_G(\sigma)$
\EndWhile
\State \textbf{return} $W_0$, $W_1$, $\sigma$, $\tau$ as in Theorem~\ref{theorem: discrete strategy improvement winning theorem}
\end{algorithmic}
\caption{Strategy Iteration}
\label{algorithm: discrete strategy improvement}
\end{algorithm}

\section{Improvement Policies} \label{section: improvement policies}
There are two major deterministic improvement policies that we consider here,
namely the \emph{locally optimizing policy} due to Jurdzi{\'n}ski and V\"oge \cite{conf/cav/VogeJ00}
and the \emph{globally optimizing policy} by Schewe \cite{conf/csl/Schewe08}.

The \emph{locally optimizing policy} $\mathcal{I}^\mathtt{loc}_G$ selects a most
profitable strategy decision in every point with respect to the current valuation.
More formally, it holds for every strategy $\sigma$, every player~0 node $v$ and
every $w \in vE$ that $w \preceq_\sigma \mathcal{I}^\mathtt{loc}_G(\sigma)(v)$.

\begin{lemma}[\cite{conf/cav/VogeJ00}]
The locally optimizing policy can be computed in polynomial time.
\end{lemma}

This policy is generally considered to be the most natural choice, particularly
because it directly corresponds to the canonical versions of strategy iteration in
related parts of game theory like discounted payoff games or simple stochastic
games. We will present a family of games on which the algorithm parameterized
with this policy requires exponentially many iterations.

The globally optimizing policy $\mathcal{I}^\mathtt{glo}_G$ on the other hand
computes a globally optimal successor strategy in the
sense that the associated valuation is the best under all allowed successor
strategies.
More formally, given a parity game $G$, an improvable strategy $\sigma$ and the
improved strategy $\sigma^* = \mathcal{I}^\mathtt{glo}_G(\sigma)$, we have for
an arbitrary strategy $\sigma'$ in the arena $\mathcal{A}_{G,\sigma}$
that $\Xi_{\sigma'} \unlhd \Xi_{\sigma^*}$.

The policy can be interpreted as providing strategy improvement with a
\emph{one-step lookahead}; it computes the optimal strategy under all possible
strategies that can be reached by a single improvement step.

The interested reader is pointed to Schewe's paper \cite{conf/csl/Schewe08} for 
all the details on how to effectively compute the optimal strategy update. 

\begin{theorem}[\cite{conf/csl/Schewe08}]
The globally optimizing policy can be computed in polynomial time.
\end{theorem}

We will also explain how to adapt the presented family of games in order to enforce
exponentially many strategy iterations on them when parameterized with the globally
optimal policy.

V\"oge mentions without proof in his thesis that there is an improvement policy that
requires at most $|V|$ many iterations to find its fixed point. We find this fact
to be quite remarkable and give a short proof of it in the following.

\begin{lemma}[\cite{voege/phd}]
Let $G$ be a parity game. There is an improvement policy $\mathcal{I}^\mathtt{lin}_G$ s.t.\ the strategy
improvement algorithm requires at most $|V|$ many iterations.
\end{lemma}

\begin{proof}
Let $G=(V,V_0,V_1,E,\Omega)$ be a parity game and let $\sigma^*$ be a $\unlhd$-optimal
strategy. We define the improvement policy $\mathcal{I}^\mathtt{lin}_G$ as follows.
\begin{displaymath}
\mathcal{I}^\mathtt{lin}_G(\sigma)(v) := \begin{cases}
\sigma^*(v) & \textrm{if } (v,\sigma^*(v)) \in \mathcal{A}_{G,\sigma} \\
\sigma(v) & \textrm{otherwise}
\end{cases}
\end{displaymath}

We will show that $\mathcal{I}^\mathtt{lin}_G$ is indeed an improvement policy and that the strategy
iteration parameterized with $\mathcal{I}^\mathtt{lin}_G$ requires at most $|V_0|$ iterations on $G$
in one go by verifying that
\begin{displaymath}
m(\sigma) \subsetneq V_0 \qquad \Longrightarrow \qquad m(\sigma) \subsetneq m(\mathcal{I}^\mathtt{lin}_G(\sigma))
\end{displaymath}
for all $\sigma$ where $m(\sigma) = \{v \in V_0 \mid \sigma(v) = \sigma^*(v)\}$.

Let $\sigma$ be a strategy s.t.\ $m(\sigma) \subsetneq V_0$. Since $m(\sigma) \subseteq m(\mathcal{I}^\mathtt{lin}_G(\sigma))$
holds by definition, we simply need to show that there is at least one node $v \in V_0$ with
$\sigma(v) \not= \sigma^*(v)$ and $(v,\sigma^*(v)) \in \mathcal{A}_{G,\sigma}$. Consider the game
$G' = (V,V_0,V_1,F,\Omega)$ where
\begin{displaymath}
F = \{(v,w) \in E \mid v \in V_1 \textrm{ or } (v,w) \in \sigma \textrm{ or } (v,w) \in \sigma^*\}
\end{displaymath}

It is easy to see that $\mathcal{A}_{G',\sigma} \subseteq \mathcal{A}_{G,\sigma}$ and also that
$\sigma^*$ is a $\unlhd$-optimal strategy w.r.t.\ $G'$. As $\sigma$ is not optimal, there must
be at least one proper improvement edge $(v,w) \in \mathcal{A}_{G',\sigma}$. By definition of $G'$,
it follows that $\sigma(v) \not= w$ and $\sigma^*(v) = w$.
\end{proof}

One may be misled to combine the existence of an improvement policy $\mathcal{I}^\mathtt{lin}_G$
that enforces at most linearly many iterations with the existence of the improvement policy
$\mathcal{I}^\mathtt{glo}_G$ that selects the optimal successor strategy in each iteration, in
order to propose that $\mathcal{I}^\mathtt{glo}_G$ should also enforce linearly many
iterations in the worst case.

The reason why this proposition is incorrect lies in the intransitivity of optimality of
strategy updates. Although it is true that $\mathcal{I}^\mathtt{lin}_G(\sigma) \unlhd \mathcal{I}^\mathtt{glo}_G(\sigma)$
for every strategy $\sigma$, this is not necessarily the case for iterated applications, i.e.\
$\mathcal{I}^\mathtt{lin}_G(\mathcal{I}^\mathtt{lin}_G(\sigma)) \unlhd \mathcal{I}^\mathtt{glo}_G(\mathcal{I}^\mathtt{glo}_G(\sigma))$
does not necessarily hold for all strategies $\sigma$.

\section{Sink Games} \label{section: sink games}
Every approach trying to construct a game family of polynomial size that requires
super-polynomially many iterations to be solved by strategy iteration (no matter
which policy the algorithm is parameterized with), needs to focus on the second
component of game valuations: there are only linearly many different values for
the first and third component while there are exponentially many for the second.

Particularly, as there are at most linearly many different cycle nodes that can occur in valuations during a run, there is no real benefit in actually using different cycle nodes. Hence our basic layout of a game exploiting exponential behavior consists of a complex structure leading to one single loop -- the only cycle node that will occur in valuations (such structures can easily be identified by preprocessing, but obviously it is not very difficult to obfuscate the whole construction without really altering its effect on the strategy iteration).
In this setting,
the strategy iteration algorithm is just improving the paths leading to the cycle node.

More formally: we call a parity game $G$ (in combination with an initial strategy $\iota_G$) a \emph{1-sink game} iff the following two properties hold:
\begin{enumerate}[(1)]
\item \emph{Sink Existence}: there is a node $v^*$ (called the \emph{1-sink} of $G$) with $v^*Ev^*$ and $\Omega(v^*) = 1$ reachable from all nodes; also, there is no other node $w$ with $\Omega(w) \leq \Omega(v^*)$.
\item \emph{Sink Seeking}: for each player~0 strategy $\sigma$ with $\Xi_{\iota_G} \unlhd \Xi_\sigma$ and each node $w$ it holds that the cycle component of $\Xi_\sigma(w)$ equals $v^*$.
\end{enumerate}

Obviously, a 1-sink game is won by player~1. Note that comparing node valuations in a 1-sink game can be reduced to comparing the path components of the respective node valuations, for two reasons. First, the cycle component remains constant. Second, the path-length component equals the cardinality of the path component, because all nodes except the sink node are more relevant than the cycle node itself.
In the case of a 1-sink game, we will therefore identify node valuations with their path component.

It is fairly easy to prove that a game is a 1-sink game indeed. One simply has to
check that the sink existence property holds by looking at the graph, that the
game is completely won by player~1, and that the 1-sink is the cycle component of
all nodes of the initial strategy.

\begin{lemma} \label{lemma: discrete strategy improvement 1 sink game lemma} Let $G$ be a parity game fulfilling the sink existence property w.r.t.\ $v^*$. $G$ is a 1-sink game iff $G$ is completely won by player~1 (i.e.\ $W_1 = V$) and for each node $w$ it holds that the cycle component of $\Xi_{\iota_G}(w)$ equals $v^*$.
\end{lemma}

\begin{proof}
The ``only-if''-part is trivial. For the ``if''-part, we need to show that the \emph{sink seeking}-property holds. Let $\sigma$ be a player~0 strategy with $\Xi_{\iota_G} \unlhd \Xi_\sigma$, $w$ be an arbitrary node and $u$ be the cycle component of $\Xi_\sigma(w)$. Due to the fact that $G$ is completely won by player~1, $u$ has to be of odd priority. Also, since $\Xi_{\iota_G} \unlhd \Xi_\sigma$, it holds that $\Omega(u) \leq \Omega(v^*)$ implying $u = v^*$ by the \emph{sink existence}-property.
\end{proof}

There is another reason why 1-sink games are an interesting subclass of parity games: we will
see later that the strategy iteration on a discounted payoff game that has been induced by the
canonic reduction from a 1-sink parity game, directly corresponds to the strategy iteration on 
the original 1-sink parity game. This connection between discounted payoff games and 1-sink
parity games allows us to directly transfer the lower bound to discounted payoff games.

\section{Lower Bound for the Locally Optimizing Policy} \label{section: locally pol}
The lower bound construction for the locally optimizing policy is a family of
1-sink parity games that implement a binary counter. In order to reduce the
overall complexity of the games, our construction relies on unbounded edge
outdegree, yielding a quadratic number of edges in total. We will discuss in the next
section how the number of edges can be reduced to a linear number and even how
to get binary outdegree.

The implementation of the binary counter is based on a structure called
\emph{simple cycles} that allows us to encode a single bit state in a given
strategy $\sigma$. By having $n$ such simple cycles, we can represent every state
of an $n$-bit binary counter. In order to allow strategy improvement the
transitions of the binary counter, we need to embed the simple cycles in a more
complicated structure called \emph{cycle gadget}, connect the cycle gadgets
of the different bits with each other, and with an additional structure called
\emph{deceleration lane}.

This section is organized as follows. First, we consider the three gadgets that will be
used in our lower bound construction, namely \emph{simple cycles}, 
the \emph{deceleration lane} and \emph{cycle gates}. Then, we present the full
construction of our lower bound
family and give a high-level description of strategy iteration on these games.
Finally, we prove that strategy iteration on the games indeed follows the
high-level description.

For the presentation of the gadgets, we assume the context of a 1-sink parity
game. The labelings and priorities of the gadgets will match the final priorities
of the lower bound family.

Gadgets consist of three kinds of nodes: \emph{input nodes}, \emph{output nodes}
and \emph{internal nodes}. Input nodes are nodes that will have incoming edges
from outside of the gadget, output nodes will have outgoing edges to the outside
of the gadget and internal nodes will not be directly connected to the outside of
the gadget.

In the context of 1-sink game $G$ and a strategy $\sigma$, we will sometimes say
that a node $v$ \emph{reaches} a node $w$ to denote the fact that $w$ lies on
the path $\pi_{\sigma,\tau_\sigma,v}$.

\subsection{Simple Cycles}

The binary counter will contain a representation of $n$ bits that are realized
by $n$ instances of a gadget called a \emph{cycle gate}. The most important part
of a cycle gate is the \emph{simple cycle} that we will introduce first. We fix
some index $i$ for the simple cycle gadget for the sake of this subsection in
order to have consistent node labelings.

A simple cycle consists of one player~0 controlled internal node $\cyclenode_i$
that is connected to a set of external nodes $D_i$ in the rest of the graph,
and one player~1 controlled input node $\cyclecenter_i$. The node $\cyclecenter_i$
itself is connected to $\cyclenode_i$ (therefore the name \emph{simple cycle}) and
to one output node $\cycleleaver_i \not\in D_i$. We note that all $\cyclecenter_i$ nodes
are the \emph{only} player~1 controlled nodes with real choices in the complete
lower bound construction.

All priorities of the simple cycle are based on an odd priority $p_i$. Intuitively,
the $p_i$ is considered to be a very small priority compared to the priorities of the
other nodes in the external graph that the simple cycle is connected to. We will
implicitly assume this in the following.

See Figure~\ref{figure: locally optimal policy simple cycle} for a simple cycle of index~$1$ with $p_1 = 3$.
The players, priorities and edges are described in Table~\ref{table: locally optimal policy simple cycle}.

\begin{table}[!h]
\begin{center}
\tabcolsep5pt
\renewcommand\arraystretch{1.2}
\begin{tabular}{l|l|c|l}
  Node & Player & Priority & Successors \\
  \hline
  $\cyclenode_i$ & $0$ & $p_i$ & $\{\cyclecenter_i\} \cup D_i$ \\
  $\cyclecenter_i$ & $1$ & $p_i+1$ & $\{\cyclenode_i,\ \cycleleaver_i\}$ \\
  $\cycleleaver_i$ & ? & $>p_i+1$ & ? \\
  $w \in D_i$ & ? & $>p_i+1$ & ? \\
\end{tabular}
\end{center}
\caption{Description of the Simple Cycle}
\label{table: locally optimal policy simple cycle}
\end{table}

\begin{figure}[!h]
\begin{center}
\fbox{
\includegraphics{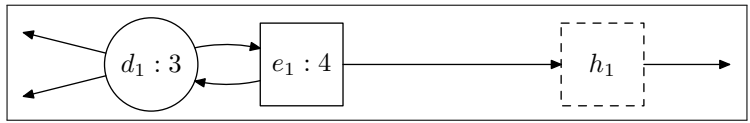}
}
\end{center}
\caption{Simple Cycle}
\label{figure: locally optimal policy simple cycle}
\end{figure}

Given a strategy $\sigma$, we say that the cycle is \emph{closed} iff
$\sigma(\cyclenode_i) = \cyclecenter_i$ and \emph{open} otherwise. A closed cycle
corresponds to a bit which is set while an open cycle corresponds to an unset bit.

The main idea now is to assign priorities to the simple cycle in such a way that
the simple cycle is won by player~0, i.e.\ the most relevant node on the cycle
needs to have an even priority. This has important consequences for the behaviour
of the player~1 controlled node.

First, assume that $\sigma(\cyclenode_i) = \cyclecenter_i$. The optimal counter-strategy
here is $\tau_\sigma(\cyclecenter_i) = \cycleleaver_i$, since otherwise player~0
would win the cycle which is impossible with $G$ being a 1-sink game. Player~0
is therefore able to \emph{force} player~1 to move out of the cycle; in other words,
\emph{setting} a bit corresponds to \emph{forcing} player~1 out of the cycle. In
a set bit, the valuation of $\cyclenode_i$ is essentially the valuation of
$\cycleleaver_i$, i.e.\ $\Xi_\sigma(\cyclenode_i) = \Xi_\sigma(\cycleleaver_i) \cup \{\cyclenode_i,\cyclecenter_i\}$.

Second, assume that $\sigma(\cyclenode_i) = w$ for some $w \in D_i$, and that
$w \prec_\sigma \cycleleaver_i$. It follows that $\cyclenode_i \prec_\sigma \cycleleaver_i$,
hence $\tau_\sigma(\cyclecenter_i) = \cyclenode_i$. The interesting part is now
that $\Xi_\sigma(\cyclecenter_i) = \Xi_\sigma(w) \cup \{\cyclenode_i,\cyclecenter_i\}$,
i.e.\ $\cyclecenter_i$ is an improving node for $\cyclenode_i$ (since
$\Xi_\sigma(w) \triangle \Xi_\sigma(\cyclecenter_i) = \{\cyclenode_i,\cyclecenter_i\}$),
but updating to $\cyclecenter_i$ would yield a much greater reward than just
$\Xi_\sigma(\cyclecenter_i)$ (namely $\Xi_\sigma(\cycleleaver_i) \cup \{\cyclecenter_i\}$
by forcing player~1 to leave the cycle).

Assume now that $w' \in D_i$ with $w \prec_\sigma w'$ but
$w' \prec_\sigma \cycleleaver_i$.
Obviously, $w'$ and $\cyclecenter_i$ are improving nodes for $\cyclenode_i$, but
$\cyclecenter_i \prec_\sigma w'$, hence by the locally improving policy, player~0
switches to $w'$ although $\cyclecenter_i$ might give a much better valuation. In
other words, by moving to $\cyclenode_i$, the player~1 node hides the fact that there
is a highly profitable node on the other side.

We formalize the behaviour of the simple cycle in two lemmas. The first
describes the valuation of $\cyclecenter_i$ depending on the state of the simple
cycle and the second explains the switching behaviour of the player~0 controlled
node. The claimed result can easily be obtained by tracing the paths that the
strategies take through the gadget, and then comparing valuations.

\begin{lemma}\label{lemma: simple cycle profitability}
Let $\sigma$ be a strategy. The following holds:
\begin{enumerate}[\em(1)]
\item\label{lemma: simple cycle profitability, item: closed closed} If cycle $i$ is closed, we have $\tau_\sigma(\cyclecenter_i) = \cycleleaver_i$.
\item\label{lemma: simple cycle profitability, item: open open} If cycle $i$ is open and $\cycleleaver_i \prec_\sigma \sigma(\cyclenode_i)$, we have $\tau_\sigma(\cyclecenter_i) = \cycleleaver_i$.
\item\label{lemma: simple cycle profitability, item: open closed} If cycle $i$ is open and $\sigma(\cyclenode_i) \prec_\sigma \cycleleaver_i$, we have $\tau_\sigma(\cyclecenter_i) = \cyclenode_i$.
\end{enumerate}
\end{lemma}

\begin{lemma}\label{lemma: simple cycle behaviour}
Let $\sigma$ be a strategy and $w = \max_{\prec_\sigma} D_i$. Let
$\sigma' = \mathcal{I}^\mathtt{loc}(\sigma)$. The following holds:
\begin{enumerate}[\em(1)]
\item\label{lemma: simple cycle behaviour, item: closed closed} If cycle $i$ is closed and $w \prec_\sigma \cycleleaver_i$, we have cycle $i$ $\sigma'$-closed (``closed cycle remains closed'').
\item\label{lemma: simple cycle behaviour, item: open open} If cycle $i$ is open, $\sigma(\cyclenode_i) \not=w$ or $\cycleleaver_i \prec_\sigma w$, we have $\sigma'(\cyclenode_i) = w$ (``open cycle remains open'').
\item\label{lemma: simple cycle behaviour, item: open close} If cycle $i$ is open, $\sigma(\cyclenode_i) = w$ and $w \prec_\sigma \cycleleaver_i$, then cycle $i$ is $\sigma'$-closed (``open cycle closes'').
\item\label{lemma: simple cycle behaviour, item: closed open} If cycle $i$ is closed and $\cycleleaver_i \prec_\sigma w$, we have $\sigma'(\cyclenode_i) = w$ (``closed cycle opens'').
\end{enumerate}
\end{lemma}

Open simple cycles have the important property that we can postpone closing them
by supplying them with new nodes $w \in D_i$ in each iteration s.t.\ $\sigma(\cyclenode_i) \prec_\sigma w$.
We will use this property in the construction of our binary counter. Since we do
not want to set all bits at the same time, rather one by one, we need to make sure
that unset bits which are not supposed to be set remain unset for some time
(more precisely, until the respective bit represents the least unset bit), and
this will be realized by this property. The device that supplies us with new
best-valued external nodes in each iteration is called \emph{deceleration lane}
and will be described next.

\subsection{Deceleration Lane}

A \emph{deceleration lane} has several, say $m$, input nodes and some
output nodes, called \emph{roots}. The lower bound construction will only
require a deceleration lane with two roots $\starteven$ and $\bitselector$, however, it would be easy to
generalize the construction of deceleration lanes to an arbitrary number of
roots.

More formally, a deceleration lane consists of $m$ 
internal nodes $\declaneodd_1$, $\ldots$, $\declaneodd_m$,
one additional internal node $\declaneroot$, $m$ input nodes $\declaneeven_1$,
$\ldots$, $\declaneeven_m$ and two output nodes $\starteven$ and $\bitselector$, called
\emph{roots} of the deceleration lane.

All priorities of the deceleration lane are based on some odd priority $p$. We
assume that all root nodes have a priority greater than $p + 2m + 1$.
See Figure~\ref{figure: locally optimal policy deceleration lane} for a
deceleration lane with $m = 6$ and $p = 15$.
The players, priorities and edges are described in Table~\ref{table: locally optimal policy deceleration lane}.

\begin{table}[!h]
\begin{center}
\tabcolsep5pt
\renewcommand\arraystretch{1.2}
\begin{tabular}{l|l|c|l}
  Node & Player & Priority & Successors \\
  \hline
  $\declaneodd_1$ & $0$ & $p$ & $\{\starteven,\ \bitselector,\ \declaneroot\}$ \\
  $\declaneodd_{i>1}$ & $0$ & $p + 2i - 2$ & $\{\starteven,\ \bitselector,\ \declaneodd_{i - 1}\}$ \\
  $\declaneroot$ & $0$ & $p + 2m + 1$ & $\{\starteven,\ \bitselector\}$ \\
  $\declaneeven_i$ & $1$ & $p + 2i - 1$ & $\{\declaneodd_i\}$ \\
  $\starteven$ & ? & $> p + 2m + 1$ & ? \\
  $\bitselector$ & ? & $> p + 2m + 1$ & ?
\end{tabular}
\end{center}
\caption{Description of the Deceleration Lane}
\label{table: locally optimal policy deceleration lane}
\end{table}

\begin{figure}[!h]
\begin{center}
\fbox{
\includegraphics{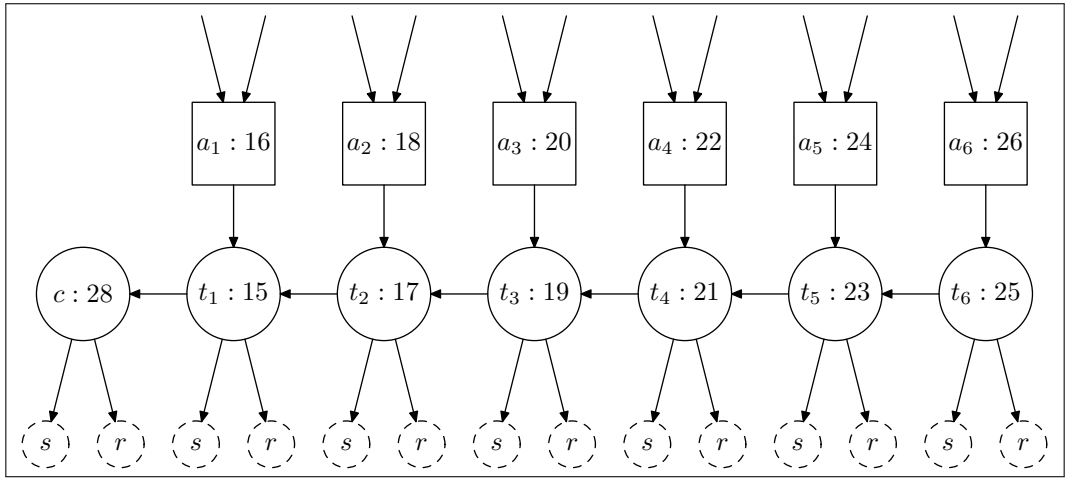}
}
\end{center}
\caption{A Deceleration Lane (with $m = 6$ and $p = 15$)}
\label{figure: locally optimal policy deceleration lane}
\end{figure}

A deceleration lane serves the following purpose.
Assume that one of the output
nodes, say $\bitselector$, has the better valuation compared to the other root
node, and
assume further that this setting sustains for some iterations.

The input nodes,
say $\declaneeven_1,\ldots,\declaneeven_m$,
now serve as an entry point, and all reach the best valued root -- $\bitselector$
-- by some internal nodes. The valuation ordering of all input nodes depends
on the iteration: at first, $\declaneeven_1$ has a better valuation than all
other input nodes. Then, $\declaneeven_2$ has a better valuation than all other
input nodes and so on.

This process continues until the other output node, say $\starteven$, has a better
valuation than $\bitselector$. Within the next iteration, the internal nodes perform a \emph{resetting}
step s.t.\ all input nodes eventually reach the new root node. One iteration after
that, $\declaneeven_1$ has the best valuation compared to all other input
nodes again.

In other words, by giving one of the roots, say $\starteven$, a better valuation
than another root, say $\bitselector$, it is possible to reset and therefore
reuse the lane again. In fact, the lower bound construction will use a deceleration
lane with two roots $\starteven$ and $\bitselector$, and will employ $\starteven$
only for resetting, i.e.\ after some iterations with $\bitselector \succ_\sigma \starteven$,
there will be one iteration with $\starteven \succ_\sigma \bitselector$ and right
after that again $\bitselector \succ_\sigma \starteven$.

From an abstract point of view, we describe the state of a deceleration lane by
which of the two roots is chosen and by how many $\declaneodd_i$ nodes are already
moving down to $\declaneroot$. Formally, we say that \emph{$\sigma$ is in
deceleration state $(x,j)$} (where $x \in \{\starteven, \bitselector\}$ and
$0<j \leq m+1$ a natural number) iff
\begin{enumerate}[(1)]
\item $\sigma(\declaneroot) = x$,
\item $\sigma(\declaneodd_1) = \declaneroot$ if $j > 1$,
\item $\sigma(\declaneodd_i) = \declaneodd_{i-1}$ for all $1 < i < j$, and
\item $\sigma(\declaneodd_i) = x$ for all $j \leq i$.
\end{enumerate}
We say that the deceleration lane is \emph{rooted in $x$} if $\sigma$ is in state
$(x,*)$, and that the \emph{index is $i$} if $\sigma$ is in state $(*,i)$.
Whenever a strategy $\sigma$ is in state $(x,i)$, we define $\dlroot(\sigma) = x$
and $\dlindex(\sigma) = i$. In this case, we say that the strategy is
\emph{well-behaved}.

We formalize the behaviour of the deceleration lane in two lemmas.
The first describes the ordering of the valuations of the input nodes depending
on the state $(x,i)$ of the deceleration lane: (1) if the ordering of the root nodes
changes, all input nodes have a worse valuation than the better root, and (2)
otherwise the best valued input node is $\declaneeven_{i-1}$.
The second explains the switching behaviour of the player~0 controlled nodes:
(1) if the ordering of the root node changes, than the whole lane resets, and
(2) otherwise the lane assembles further, providing a new best-valued input
node.

\begin{lemma}\label{lemma: deceleration lane valuations}
Let $\sigma$ be a strategy in deceleration state $(x,i)$. Let $\bar x$
denote the other root. Then
\begin{enumerate}[\em(1)]
\item $x \prec_\sigma \bar x$ implies $\declaneeven_j \prec_\sigma \bar x$ for all $j$
      (``resetting results in unprofitable lane'').
\item\label{lemma: deceleration lane valuations, item: standard ordering} $\bar x \prec_\sigma x$ implies
	  $x \prec_\sigma \declaneeven_i \prec_\sigma \ldots \prec_\sigma \declaneeven_m \prec_\sigma \declaneroot \prec_\sigma \declaneeven_1 \prec_\sigma \ldots \prec_\sigma \declaneeven_{i-1}$
	  (``new best-valued node in each iteration'').
\end{enumerate}
\end{lemma}

\begin{lemma}\label{lemma: deceleration lane behaviour}
Let $\sigma$ be a strategy that is in deceleration state $(x,i)$. Let $\bar x$
denote the other root. Let $\sigma' = \mathcal{I}^\mathtt{loc}(\sigma)$. Then
\begin{enumerate}[\em(1)]
\item\label{lemma: deceleration lane behaviour, item: changing} $x \prec_\sigma \bar x$ implies that $\sigma'$ is in state $(\bar x,1)$
      (``lane resets'').
\item\label{lemma: deceleration lane behaviour, item: assembling} $\bar x \prec_\sigma x$ implies that $\sigma'$ is in state $(x,\min(i,m)+1)$
      (``lane assembles one step at a time'').
\item\label{lemma: deceleration lane behaviour, item: well behaved} $\sigma'$ is well-behaved (``always ending up with well-behaved strategies'').	  
\end{enumerate}
\end{lemma}

The main purpose of a deceleration lane is to absorb the update activity of
other nodes in such a way that wise (i.e.\ edges that will result in much
better valuations \emph{after} switching and reevaluating)
strategy updates are postponed. 
Consider a
node for instance that has more than one proper improving switch; the locally
optimizing policy will select the edge with the best valuation to be switched.
In order to prevent that one particular improving switch is applied for some
iterations, one can connect the node to the input nodes of the deceleration lane.

The particular scenario in which we will use the deceleration lane are
\emph{simple cycles} as described in the previous subsection. 
We will connect the simple cycles encoding the bits of our counter to
the deceleration lane in such a way, that lower cycles have less edges entering
the deceleration lane. This construction ensures that lower open cycles
(representing unset bits) will close (i.e.\ set the corresponding bit)
before higher open cycles (representing higher unset bits) have their turn to close.

\subsection{Cycle Gate}

The simple cycles will appear in a more complicated gadget, called
\emph{cycle gate}. We will have $n$ different cycle gates in the game number $n$
of the lower bound family, hence we fix some index $i$ for the cycle gate
gadget for the sake of this subsection in order to have consistent node labelings.

Formally,
a cycle gate consists of two internal nodes $\cyclecenter_i$ and $\cycleleaver_i$,
two input nodes $\cycleaccess_i$ and $\cycleselector_i$, and two output nodes
$\cyclenode_i$ and $\upperselector_i$. The output node $\cyclenode_i$ will be
connected to a set of other nodes $D_i$ in the game graph, and $\upperselector_i$
to some other set $K_i$ as well. The two nodes $\cyclenode_i$ and $\cyclecenter_i$
form a simple cycle as described earlier.

All priorities of the cycle gate are based on two odd priorities $p_i$ and
$p_i'$.
See Figure~\ref{figure: locally optimal policy cycle gate} for a cycle gate of index~$1$ with $p_1' = 3$
and $p_1 = 33$.
The players, priorities and edges are described in Table~\ref{table: locally optimal policy cycle gate}.

\begin{table}[!h]
\begin{center}
\tabcolsep5pt
\renewcommand\arraystretch{1.2}
\begin{tabular}{l|l|c|l}
  Node & Player & Priority & Successors \\
  \hline
  $\cyclenode_i$ & $0$ & $p_i'$ & $\{\cyclecenter_i\} \cup D_i$ \\
  $\cyclecenter_i$ & $1$ & $p_i'+1$ & $\{\cyclenode_i,\ \cycleleaver_i\}$ \\
  $\cycleselector_i$ & $0$ & $p_i'+3$ & $\{\cycleaccess_i,\ \upperselector_i\}$ \\
  $\upperselector_i$ & $0$ & $p_i$ & $K_i$ \\
  $\cycleaccess_i$ & $1$ & $p_i+2$ & $\{\cyclecenter_i\}$ \\
  $\cycleleaver_i$ & $1$ & $p_i+3$ & $\{\upperselector_i\}$
\end{tabular}
\end{center}
\caption{Description of the Cycle Gate}
\label{table: locally optimal policy cycle gate}
\end{table}

\begin{figure}[!h]
\begin{center}
\fbox{
\includegraphics{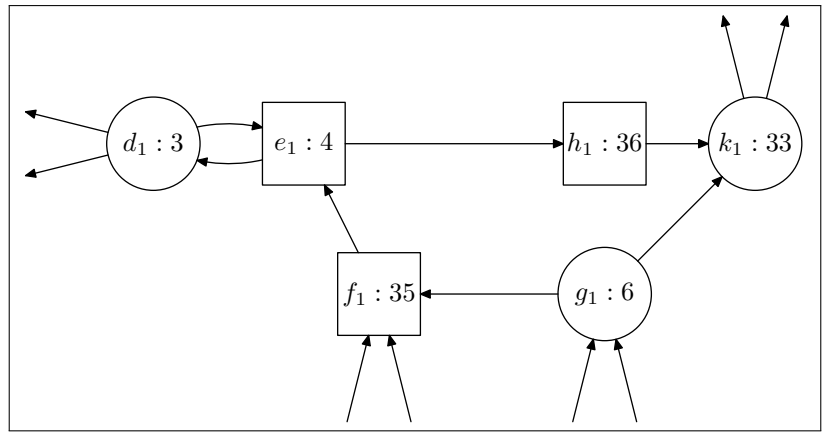}
}
\end{center}
\caption{A Cycle Gate (index $1$ with $p'_1 = 3$ and $p_1 = 33$)}
\label{figure: locally optimal policy cycle gate}
\end{figure}

The main idea behind a cycle gate is to have a pass-through structure
\emph{controlled by the simple cycle} that is
either very profitable or quite unprofitable. The pass-through structure of the
cycle gate has one major input node, named $\cycleselector_i$, and one major
output node, named $\upperselector_i$. The input node is controlled by player~0
and connected via two paths with the output node; there is a direct edge and a
longer path leading through the interior of the cycle gate.

However, the longer path only leads to the output node if the simple cycle,
consisting of one player~0
node $\cyclenode_i$ and one player~1 node $\cyclecenter_i$, is \emph{closed}.
In this case, it is possible and profitable to
reach the output node via the internal path; otherwise, this path is not
accessible, and hence, the input node has to select the unprofitable direct way
to reach the output node.

We will have one additional input node, named $\cycleaccess_i$, that can only
access the path leading through the interior of the cycle gate, for the following
purpose. Assume that the simple cycle has just been closed and now the path leading
through the interior becomes highly profitable. Hence, the next switching event
to happen will be the node $\cycleselector_i$ switching from the direct path to
the path through the interior. However, it will be useful to be able to reach
the highly profitable path from some parts of the outside graph one iteration before
it is accessible via $\cycleselector_i$. For this reason, we include an additional
input node $\cycleaccess_i$ that immediately accesses the interior path.

We say that a cycle gate is \emph{closed} resp.\ \emph{open} iff the interior
simple cycle is closed resp.\ open. Similarly, we say that a cycle gate is
\emph{accessed} resp.\ \emph{skipped} iff the access control node $\cycleselector_i$
moves through the interior ($\sigma(\cycleselector_i) = \cycleaccess_i$) resp.\
directly to $\upperselector_i$.

From an abstract point of view, we describe the state of a cycle gate by a 
pair $(\beta_i(\sigma),\alpha_i(\sigma)) \in \{0,1\}^2$. The first component
describes the state of the
simple cycle, and the second component gives the state of the access control node.
Formally, we have the following.
\begin{enumerate}[(1)]
\item $\beta_i(\sigma) = 1$ iff the $i$-th cycle gate is closed, and
\item $\alpha_i(\sigma) = 1$ iff the $i$-th cycle gate is accessed.
\end{enumerate}

We formalize the behaviour of the cycle gate in two lemmas.
The first describes the valuation of all important nodes of the cycle gate, using
our knowledge of simple cycles of Lemma~\ref{lemma: simple cycle profitability}.
The second explains the switching behaviour of the access control node. The
behaviour of the simple cycle contained in the cycle gate is described by
Lemma~\ref{lemma: simple cycle behaviour}.

\begin{lemma}\label{lemma: cycle gate profitability}
Let $\sigma$ be a strategy. Then
\begin{enumerate}[\em(1)]
\item If gate $i$ is open, we have $\cycleaccess_i \prec_\sigma \sigma(\cyclenode_i)$.
\item If gate $i$ is closed, we have $\sigma(\upperselector_i) \prec_\sigma \cycleaccess_i$.
\item If gate $i$ is closed and skipped, we have $\cycleselector_i \prec_\sigma \cycleaccess_i$.
\item If gate $i$ is accessed, we have $\cycleaccess_i \prec_\sigma \cycleselector_i$.
\item If gate $i$ is skipped, we have $\sigma(\upperselector_i) \prec_\sigma \cycleselector_i$.
\end{enumerate}
\end{lemma}

\begin{lemma}\label{lemma: cycle gate access behaviour}
Let $\sigma$ be a strategy and $\sigma' = \mathcal{I}^\mathtt{loc}(\sigma)$.
\begin{enumerate}[\em(1)]
\item\label{lemma: cycle gate access behaviour, item: closed accessed} If gate $i$ is $\sigma$-closed, then gate $i$ is $\sigma'$-accessed
      (``closed gates will be accessed'').
\item\label{lemma: cycle gate access behaviour, item: open skipped} If gate $i$ is $\sigma$-open and $\sigma(\cyclenode_i) \prec_\sigma \cycleleaver_i$, then gate $i$ is $\sigma'$-skipped
      (``open gates with unprofitable exit nodes will be skipped'').
\item If gate $i$ is $\sigma$-open and $\cycleleaver_i \prec_\sigma \sigma(\cyclenode_i)$, then gate $i$ is $\sigma'$-accessed
      (``open gates with profitable exit nodes will be accessed'').
\end{enumerate}
\end{lemma}

\noindent The last two items of Lemma~\ref{lemma: cycle gate access behaviour} are based on
the uniqueness of priorities in the game, implying that there are no priorities
between $\cycleaccess_i$ and $\cycleleaver_i$.

We will use cycle gates to represent the bit states of a binary counter: unset
bits will correspond to cycle gates with the state $(0,0)$,
set bits to the state $(1,1)$. Setting and resetting bits
therefore traverses more than one phase, more precisely, from $(0,0)$
over $(1,0)$ to $(1,1)$, and from the
latter again over $(0,1)$ to $(0,0)$.
Particularly, it can be observed that the second component of the cycle gate states
switches one iteration after the first component in both cases.

\subsection{Lower Bound Construction}

In this subsection, we provide the complete construction of the lower bound
family. It essentially consists of a 1-sink $\finalcycle$, a deceleration lane
of length $2n$ that is connected to the two roots $\starteven$ and $\bitselector$,
and $n$ cycle gates. The simple cycles of the cycle gates are connected to
the roots and to the deceleration lane with the important detail, that lower
cycle gates have less edges to the deceleration lane. This construction ensures
that lower open cycle gates will close before higher open cycle gates.

The output node of a cycle gate is connected to the 1-sink and to the
$\cycleselector_*$-input nodes of all higher cycle gates. The $\starteven$ root
node is connected to all $\cycleaccess_*$-input nodes, the $\bitselector$
root node is connected to all $\cycleselector_*$-input nodes.

We now give the formal construction. 
The games are denoted by $G_n=(V_n,V_{n,0},V_{n,1},E_n,\Omega_n)$.
The sets of nodes are
\begin{displaymath}
V_n := \{\finalcycle, \starteven, \declaneroot, \bitselector\} \cup \{\declaneodd_i, \declaneeven_i \mid 1 \leq i \leq 2n\} \cup \{\cyclenode_i, \cyclecenter_i, \cycleselector_i, \upperselector_i, \cycleaccess_i, \cycleleaver_i \mid 1 \leq i \leq n\}
\end{displaymath}
The players, priorities and edges are described in Table~\ref{table: locally optimizing lower bound game}. The game $G_3$ is
depicted in Figure~\ref{figure: locally optimizing lower bound game}.

\begin{table}[!h]
\begin{center}
\tabcolsep5pt
\renewcommand\arraystretch{1.2}
\begin{tabular}{l|l|c|l}
  Node & Player & Priority & Successors \\
  \hline
  $\declaneodd_1$ & $0$ & $4n + 3$ & $\{\starteven,\ \bitselector,\ \declaneroot\}$ \\
  $\declaneodd_{i>1}$ & $0$ & $4n + 2i + 1$ & $\{\starteven,\ \bitselector,\ \declaneodd_{i - 1}\}$ \\
  $\declaneeven_i$ & $1$ & $4n + 2i + 2$ & $\{\declaneodd_i\}$ \\
  $\declaneroot$ & $0$ & $8n + 4$ & $\{\starteven,\ \bitselector\}$ \\
  \hline
  $\cyclenode_i$ & $0$ & $4i + 1$ & $\{\starteven,\ \cyclecenter_i,\ \bitselector\} \cup \{\declaneeven_j \mid j < 2i+1\}$ \\
  $\cyclecenter_i$ & $1$ & $4i + 2$ & $\{\cyclenode_i,\ \cycleleaver_i\}$ \\
  $\cycleselector_i$ & $0$ & $4i + 4$ & $\{\cycleaccess_i,\ \upperselector_i\}$ \\
  $\upperselector_i$ & $0$ & $8n + 4i + 7$ & $\{\finalcycle\} \cup \{\cycleselector_{j} \mid i < j \leq n\}$ \\
  $\cycleaccess_i$ & $1$ & $8n + 4i + 9$ & $\{\cyclecenter_i\}$ \\
  $\cycleleaver_i$ & $1$ & $8n + 4i + 10$ & $\{\upperselector_i\}$ \\
  \hline
  $\starteven$ & $0$ & $8n + 6$ & $\{\cycleaccess_j \mid j \leq n\} \cup \{\finalcycle\}$ \\
  $\bitselector$ & $0$ & $8n + 8$ & $\{\cycleselector_j \mid j \leq n\} \cup \{\finalcycle\}$ \\
  $\finalcycle$ & $1$ & $1$ & $\{\finalcycle\}$
\end{tabular}
\end{center}
\caption{Lower Bound Construction for the Locally Optimizing Policy}
\label{table: locally optimizing lower bound game}
\end{table}

\begin{figure}[!h]
\begin{center}
\rotatebox{90}{
\scalebox{0.88}{
\includegraphics{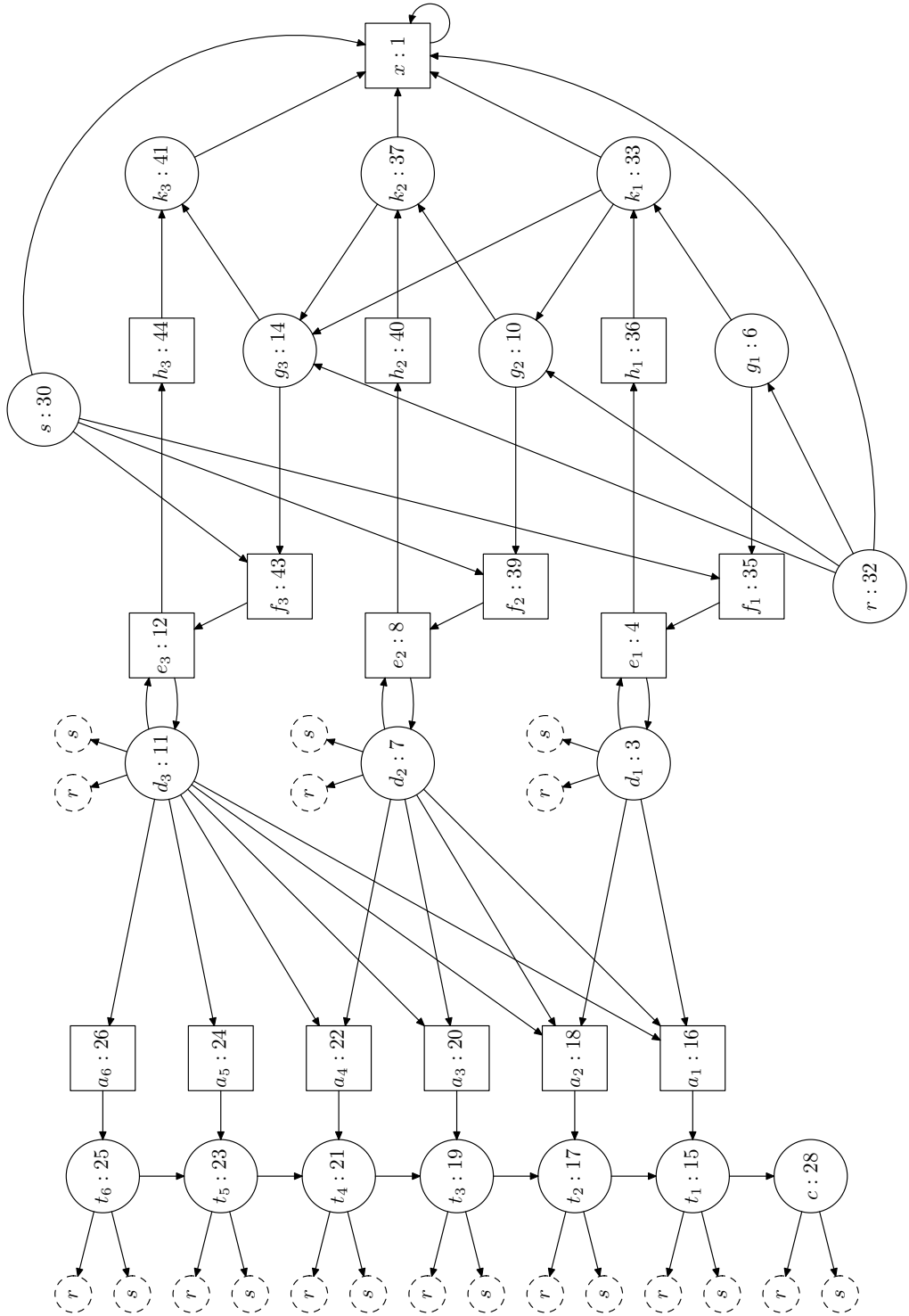}
}
}
\end{center}
\caption{Locally Optimizing Lower Bound Game $G_3$}
\label{figure: locally optimizing lower bound game}
\end{figure}

\begin{fact}
The game $G_n$ has $10 \cdot n + 4$ nodes, $1.5 \cdot n^2 + 20.5 \cdot n + 5$ edges
and $12 \cdot n + 8$ as highest priority. In particular, $|G_n| = \BigO{n^2}$.
\end{fact}

As an initial strategy we select the following $\iota_{G_n}$. It will
correspond to the global counter state in which no bit has been set.
\begin{displaymath} \qquad
\iota_{G_n}(\declaneodd_1) = \declaneroot
\quad
\iota_{G_n}(\cycleselector_i) = \upperselector_i
\quad
\iota_{G_n}(* \in \{\declaneodd_{i > 1},\declaneroot,\cyclenode_i\}) = \bitselector
\quad
\iota_{G_n}(* \in \{\upperselector_i,\starteven,\bitselector\}) = \finalcycle
\end{displaymath}
Note that $\iota_{G_n}$ particularly is well-behaved. Hence, by Lemma~\ref{lemma: deceleration lane behaviour}(\ref{lemma: deceleration lane behaviour, item: well behaved})
we know that all strategies that will occur in a run of the strategy improvement
algorithm will be well-behaved.

We will see in the next section, how the family $G_n$ can be refined in such a
way that it only comprises a linear number of edges. The reason why we present
the games with a quadratic number of edges first is that the refined family
looks even more confusing and obfuscates the general principle.

\begin{lemma}
Let $n > 0$.
\begin{enumerate}[\em(1)]
\item The game $G_n$ is completely won by player 1.
\item $\finalcycle$ is the 1-sink of $G_n$ and the cycle component of $\Xi_{\iota_{G_n}}(w)$ equals $\finalcycle$ for all $w$.
\end{enumerate}
\end{lemma}

\proof
Let $n > 0$.
\begin{enumerate}[(1)]
\item Note that the only nodes owned by player~1 with an outdegree greater than 1 are $\cyclecenter_1$,$\ldots$,$\cyclecenter_n$. Consider the player~1 strategy $\tau$ which selects to move to $\cycleleaver_i$ from $\cyclecenter_i$ for all $i$. Now it is the case that $G_n|_\tau$ contains exactly one cycle that is eventually reached no matter what player~0 does, namely the self-cycle at $\finalcycle$ which is won by player~1.
\item The self-cycle at $\finalcycle$ obviously is the $1$-sink since it can be reached from all other nodes and has the smallest priority $1$. Since $\finalcycle E\finalcycle$ is the only cycle won by player~1 in $G_n|_{\iota_{G_n}}$, $\finalcycle$ must be the cycle component of each node valuation w.r.t.\ $\iota_{G_n}$.\qed
\end{enumerate}\vspace{2 pt}

\noindent By Lemma~\ref{lemma: discrete strategy improvement 1 sink game lemma} it follows that $G_n$ is a 1-sink game, hence
it is safe to identify the valuation of a node with its path component from now on.

\subsection{Lower Bound Description and Phases}

Here, we describe how the binary counter performs the task of
counting by strategy improvement. Our
games implement a full binary counter in which every bit is represented by a
simple cycle encapsulated in a cycle gate. An unset bit $i$ corresponds to an open
simple cycle in cycle gate $i$, a set bit $i$ corresponds to a closed simple cycle in
cycle gate $i$.

Formally, we represent the bit state of the counter by elements from
$B_n = \{0,1\}^n$. For $b=(b_n,\ldots,b_1) \in B_n$, let $b_i$ denote
the $i$-th component in $b$ for every $i \leq n$, where $b_n$ denotes the most and $b_1$
denotes the least significant bit. By $b \oplus 1$, we denote the increment of
the number represented by $b$ by $1$. The least resp.\ greatest bit states are denoted
by $\boldzero_n$ resp.\ $\boldone_n$.
We refer to the least unset bit by $\mu b := \min (\{n+1\} \cup \{i \leq n \mid b_i = 0\})$,
and similarly to the least set bit by $\nu b := \min (\{n+1\} \cup \{i \leq n \mid b_i = 1\})$.

From the most abstract point of view, our lower bound construction performs
counting on $B_n$. However, the increment of a global bit state requires more than
one strategy iteration, more precisely four different phases that will be described next
(with one phase of dynamic length).


Every phase is defined w.r.t.\ a given global counter state $b \in B_n$.
Let $b \in B_n$ be a global bit state different from $\boldone_n$.

An abstract counter performs the increment from $b$ to $b \oplus 1$ by
computing $b[\mu b\mapsto 1][j{<}\mu b\mapsto 0]$, i.e.\ by setting bit $\mu b$
and by resetting all lower bits $j{<}\mu b$.
In the context of the games, we start in phase 1 corresponding to $b$, and then
proceed to phase 2 and phase 3 corresponding to $b[\mu b\mapsto 1]$, from
phase 3 to phase 4 corresponding to $b[\mu b\mapsto 1][j<\mu b\mapsto 0]$, and
finally from phase 4 to phase 1 again. The transition from phase 2 to phase 3
and from phase 4 to phase 1 handles the correction of the internal structure
connecting the cycles with each other.


For the sake of this subsection, let $\sigma$ be a strategy and $b \in B_n$
be a global counter state. All phases will be defined
w.r.t.\ $\sigma$ and $b \in B_n$. Let $\sigma' = \mathcal{I}^\mathtt{loc}(\sigma)$.

To keep everything as simple as possible and to be able to prove all the lemmas without
considering special cases, we will assume that $b$ is different
from $\boldzero_n$ and that the two highest bits in $b$ are zero and remain zero,
i.e.\ we will only use the first $n-2$ bits for counting. Note however, that every bit
works as intended in the counter.

Given a strategy $\sigma$, we denote the \emph{associated simple cycle state}
$(\beta_n(\sigma),\ldots,\beta_1(\sigma))$ by $b_\sigma$, and the
\emph{associated access state} $(\alpha_n(\sigma),\ldots,\alpha_1(\sigma))$ by
$a_\sigma$.

Recall that every strategy $\sigma$ occurring will be well-behaved.
In addition to the deceleration lane and the cycle gates, we have two more
structures that are controlled by a strategy $\sigma$, namely the two roots
$\bitselector$ and $\starteven$, and the cycle gate output nodes $\upperselector_i$.
We write $\sigma(\bitselector) = i$ to denote that $\sigma(\bitselector) = \cycleselector_i$,
and $\sigma(\bitselector) = n+1$ if $\sigma(\bitselector) = \finalcycle$;
we write $\sigma(\starteven) = i$ to denote that $\sigma(\starteven) = \cycleaccess_i$,
and $\sigma(\starteven) = n+1$ if $\sigma(\starteven) = \finalcycle$; we write
$\sigma(\upperselector_i) = j$ to denote that $\sigma(\upperselector_i) = \cycleselector_j$,
and $\sigma(\upperselector_i) = n+1$ if $\sigma(\upperselector_i) = \finalcycle$.
We also use a more compact notation for the strategy decision of $\cyclenode_i$-nodes
of open cycles. We write $\sigma(\cyclenode_i) = j$ if $\sigma(\cyclenode_i) = \declaneeven_j$.

Recall that we say that a strategy $\sigma$ is \emph{rooted} in $\starteven$ or $\bitselector$,
if every path in the deceleration lane conforming to $\sigma$ eventually exits to $\starteven$
resp.\ $\bitselector$. Likewise, we say that $\sigma$ has index $i$ if all nodes of the
deceleration lane with smaller index $j<i$ are moving down the lane by $\sigma$, and that
$i$ is the first index which is directly exiting through the root.

The first phase, called the \emph{waiting} phase, corresponds to a stable
strategy $\sigma$ in which open cycles are busy waiting to be closed while
the deceleration lane is assembling. Cycle gates that correspond to set bits are 
closed and accessed, while cycle gates of unset bits are open and skipped, i.e.\
$b = b_\sigma = a_\sigma$. The
\emph{selector nodes} $\upperselector_i$ move to the next higher cycle gate
corresponding to a set bit, and both roots are connected to the least set bit
$\nu b$.

\noindent More formally, we say that $\sigma$ is a $b$-phase 1 strategy iff all the following
conditions hold:
\begin{enumerate}[(1)]
\item\label{phase 1 condition bits} $b = b_\sigma = a_\sigma$, i.e.\ set bits correspond to closed and accessed cycle gates, while unset bits correspond to open and skipped cycle gates,
\item\label{phase 1 deceleration root} $\dlroot(\sigma) = \bitselector$, i.e.\ the strategy is rooted in $\bitselector$,
\item\label{phase 1 root connection} $\sigma(\starteven) = \sigma(\bitselector) = \nu b$, i.e.\ both roots are connected to the least set bit,
\item\label{phase 1 upper selector} $\sigma(\upperselector_i) = \min (\{j > i \mid b_j = 1\} \cup \{n+1\})$, i.e.\ the selector nodes move to the next set bit,
\item\label{phase 1 condition index} $\dlindex(\sigma) \leq 2\mu b+2$, i.e.\ the deceleration lane has not passed the least unset bit, and
\item\label{phase 1 condition open not optimal} $\sigma(\cyclenode_j) \not= \dlindex(\sigma) - 1$ for all $j$ with $b_j = 0$, i.e.\ every open cycle node is not connected to the best-valued node of the lane.
\end{enumerate}\vspace{2 pt}

\noindent The only improving switches in the first phase are edges of open simple cycles
and edges of the deceleration lane.

\begin{lemma}\label{lemma: from phase 1 to phase 1}
Let $\sigma$ be a $b$-phase 1 strategy with $\dlindex(\sigma) < 2\mu b+2$.
Then $\sigma'$ is a $b$-phase 1 strategy with $\dlindex(\sigma') = \dlindex(\sigma) + 1$,
and if $\dlindex(\sigma) > 1$, then $\sigma'(\cyclenode_{\mu b}) = \dlindex(\sigma) - 1$.
\end{lemma}

\begin{proof}
Let $\sigma$ be a $b$-phase 1 strategy, $\Xi := \Xi_\sigma$ and
$\dlindex(\sigma) < 2\mu b+2$.

We first compute the valuations
for all those nodes directly that do not involve any complicated strategy
decision of player~1. Obviously, $\Xi(\finalcycle) = \emptyset$.
By Lemma~\ref{lemma: simple cycle profitability}(\ref{lemma: simple cycle profitability, item: closed closed})
we know that for all set bits $i$ (i.e.\ $b_i = 1$) we have the following.
\begin{align*}
\Xi(\cyclecenter_i) = \{\cyclecenter_i\} \cup \Xi(\cycleleaver_i) & &
\Xi(\cyclenode_i) = \{\cyclecenter_i,\cyclenode_i\} \cup \Xi(\cycleleaver_i) & &
\Xi(\cycleaccess_i) = \{\cyclecenter_i,\cycleaccess_i\} \cup \Xi(\cycleleaver_i)
\end{align*}
Using these equations, we are able to compute many other valuations that do not involve any complicated strategy decision of player~1.
Let $U_j = \{\cycleselector_j, \cycleaccess_j, \cyclecenter_j, \cycleleaver_j, \upperselector_j\}$. The following holds (by $\CF_p(A)$ we denote the function that returns $A$ if $p$ holds and $\emptyset$ otherwise):
\begin{align*}
\Xi(\upperselector_i) &= \{\upperselector_i\} \cup \bigcup \{U_j \mid j{>}i,b_j{=}1\} &
\Xi(\cycleleaver_i)   &= \{\cycleleaver_i, \upperselector_i\} \cup \bigcup \{U_j \mid j{>}i,b_j{=}1\} \\
\Xi(\cycleselector_i) &= \{\cycleselector_i,\upperselector_i\} \cup \bigcup \{U_j \mid j \geq i,b_j=1\} &
\Xi(\bitselector)     &= \{\bitselector\} \cup \bigcup \{U_j \mid b_j=1\} \\
\Xi(\starteven)       &= \{\starteven\} \cup \CF_{b\not=\boldzero_n}(\bigcup \{U_j \mid b_j=1\} \setminus \{\cycleselector_{\nu b}\}) &
\Xi(\declaneroot)     &= \{\declaneroot, \bitselector\} \cup \bigcup \{U_j \mid b_j=1\} \\
\Xi(\declaneodd_i)    &= \{\declaneodd_i\} \cup \Xi(\bitselector) \cup \CF_{i < \dlindex(\sigma)}(\{\declaneodd_j \mid j < i\} \cup \{\declaneroot\}) &
\Xi(\declaneeven_i)   &= \{\declaneeven_i\} \cup \Xi(\declaneodd_i)
\end{align*}\vspace{-2 pt}

\noindent It is easy to see that we have the following orderings on the nodes specified above.
\begin{equation}\tag{a}\label{internal: eq 1}
\starteven \prec_\sigma \bitselector \prec_\sigma \declaneeven_* \prec_\sigma \cycleleaver_*
\end{equation}
By Lemma~\ref{lemma: simple cycle profitability}(\ref{lemma: simple cycle profitability, item: open open}), it follows from (\ref{internal: eq 1})
that $\tau_\sigma(\cyclecenter_i) = \cyclenode_i$
for all unset bits $i$ (i.e.\ $b_i = 0$), hence we are able to compute the
valuations of the remaining nodes.
\begin{align*}
\Xi(\cyclecenter_i) = \{\cyclecenter_i\} \cup \Xi(\cyclenode_i) & &
\Xi(\cycleaccess_i) = \{\cyclecenter_i,\cycleaccess_i\} \cup \Xi(\cyclenode_i)
\end{align*}
This completes the valuation of $\Xi$ for all nodes.

It is easy to see that for every $i$ with $b_i = 0$ and every $j$ with
$b_j = 1$ s.t.\ there is no $i < i' < j$ with $b_{i'} = 1$, the following holds:
\begin{equation}\tag{b}\label{internal: eq 2}
\cycleaccess_i \prec_\sigma \cycleaccess_j \qquad
\cycleselector_i \prec_\sigma \cycleselector_j
\end{equation}
Also, for $i > j$ with $b_i = 1$ and $b_j = 1$ we have
\begin{equation}\tag{c}\label{internal: eq 3}
\cycleaccess_i \prec_\sigma \cycleaccess_j \qquad
\cycleselector_i \prec_\sigma \cycleselector_j
\end{equation}
By (\ref{internal: eq 1}) and Lemma~\ref{lemma: deceleration lane valuations}(\ref{lemma: deceleration lane valuations, item: standard ordering}) we obtain that the following holds:
\begin{equation}\tag{d}\label{internal: eq 4}
\declaneeven_{\dlindex(\sigma)} \prec_\sigma \ldots \prec_\sigma \declaneeven_{2n} \prec_\sigma \declaneeven_1 \prec_\sigma \ldots \prec_\sigma \declaneeven_{\dlindex(\sigma)-1}
\end{equation}
We are now ready to prove that $\sigma'$ is of the desired form.
\begin{enumerate}[(1)]
\item[(\ref{phase 1 condition bits})] By Lemma~\ref{lemma: simple cycle behaviour}(\ref{lemma: simple cycle behaviour, item: closed closed}) and (\ref{internal: eq 1}) we derive that closed cycles remain closed.
By Lemma~\ref{lemma: cycle gate access behaviour}(\ref{lemma: cycle gate access behaviour, item: closed accessed}) we derive that closed cycles remain accessed.
By (\ref{internal: eq 1}) and
Lemma~\ref{lemma: cycle gate access behaviour}(\ref{lemma: cycle gate access behaviour, item: open skipped}) we derive that open cycles remain skipped.

By phase 1 condition (\ref{phase 1 condition index}), phase 1 condition
(\ref{phase 1 condition open not optimal}), (\ref{internal: eq 4}), it follows
that for every $j$ with $b_j = 0$, there is an improving node $\declaneeven_*$
for $\cyclenode_j$. By Lemma~\ref{lemma: simple cycle behaviour}(\ref{lemma: simple cycle behaviour, item: open open}), we conclude that open cycles remain open.

\item[(\ref{phase 1 deceleration root})] By (\ref{internal: eq 1}) and
Lemma~\ref{lemma: deceleration lane behaviour}(\ref{lemma: deceleration lane behaviour, item: assembling}).

\item[(\ref{phase 1 root connection})] By (\ref{internal: eq 2}) and (\ref{internal: eq 3}).

\item[(\ref{phase 1 upper selector})] By (\ref{internal: eq 2}) and (\ref{internal: eq 3}).

\item[(\ref{phase 1 condition index})] By Lemma~\ref{lemma: deceleration lane behaviour}(\ref{lemma: deceleration lane behaviour, item: assembling}).

\item[(\ref{phase 1 condition open not optimal})] By Lemma~\ref{lemma: deceleration lane behaviour}(\ref{lemma: deceleration lane behaviour, item: assembling}) and
Lemma~\ref{lemma: deceleration lane valuations}(\ref{lemma: deceleration lane valuations, item: standard ordering}).
\end{enumerate}

By Lemma~\ref{lemma: deceleration lane behaviour}(\ref{lemma: deceleration lane behaviour, item: assembling})
it follows that $\dlindex(\sigma') = \dlindex(\sigma) + 1$.

If $\dlindex(\sigma) > 1$, then we have by (\ref{internal: eq 1}) and (\ref{internal: eq 4}) that
$\sigma'(\cyclenode_{\mu b}) = \dlindex(\sigma) - 1)$.
\end{proof}

The first phase ends, when a simple cycle corresponding to an unset bit
has no more edges leading to the deceleration lane that keeps it busy waiting,
and closes.
Since lower bits have less edges going to the lane, it is clear that this will
be the least unset bit $\mu b$.

The second phase, called the \emph{set} phase, corresponds to a strategy $\sigma$
in which the least unset bit has just been set, i.e.\ to the global state
$b[\mu b\mapsto 1] = b_\sigma$. The selector nodes and roots are as in phase 1
and also the access states, i.e.\ $b = a_\sigma$.

\noindent More formally, we say that $\sigma$ is a $b$-phase 2 strategy iff all the following conditions hold:
\begin{enumerate}[(1)]
\item $b[\mu b\mapsto 1] = b_\sigma$ and $b = a_\sigma$, i.e.\ set bits correspond to closed and accessed (for all set bits except for $\mu b$) cycle gates, while unset bits correspond to open and skipped cycle gates,
\item $\dlroot(\sigma) = \bitselector$, i.e.\ the strategy is rooted in $\bitselector$,
\item $\sigma(\starteven) = \sigma(\bitselector) = \nu b$, i.e.\ both roots are connected to the former least set bit,
\item $\sigma(\upperselector_i) = \min (\{j > i \mid b_j = 1\} \cup \{n+1\})$, i.e.\ the selector nodes move to the next set bit,
\item $\dlindex(\sigma) \leq 2\mu b+3$, i.e.\ the deceleration lane has not passed the next bit, and
\item $\sigma(\cyclenode_j) \not= \dlindex(\sigma) - 1$ for all $j > \mu b$ with $b_j = 0$, i.e.\ every higher open cycle node is not connected to the best-valued node of the lane.
\end{enumerate}

\begin{lemma}\label{lemma: from phase 1 to phase 2}
Let $\sigma$ be a $b$-phase 1 strategy with $\dlindex(\sigma) = 2\mu b+2$ and
$\sigma(\cyclenode_{\mu b}) = \dlindex(\sigma)$.
Then $\sigma'$ is a $b$-phase 2 strategy.
\end{lemma}

\begin{proof}
This can be shown essentially the same way as Lemma~\ref{lemma: from phase 1 to phase 1}.
The only difference now is that $\cyclenode_{\mu b}$ has no more improving switches
to the deceleration lane and hence, by Lemma~\ref{lemma: simple cycle behaviour}(\ref{lemma: simple cycle behaviour, item: open close}),
we learn that the $\mu b$-cycle has to close.
\end{proof}

In phase 2, the deceleration lane is still assembling, and the improving switches again
include edges of open simple cycles and edges of the deceleration lane.
Additionally, it is improving for the cycle gate $\mu b$ to be accessed and
for the root $\starteven$ to update to cycle gate $\mu b$. By performing all
these switches, we enter phase three.

The third phase, called the \emph{access} phase, is defined by a renewed
correspondence of the cycle gate structure again, i.e.\
$b[\mu b\mapsto 1] = b_\sigma = a_\sigma$. The $\starteven$ root is connected
to $\mu b$ while $\bitselector$ is still connected to $\nu b$.
This implies
that $\starteven$ now has a much better valuation than $\bitselector$.

\noindent More formally, we say that $\sigma$ is a $b$-phase 3 strategy iff all the following conditions hold:
\begin{enumerate}[(1)]
\item $b[\mu b\mapsto 1] = b_\sigma = a_\sigma$, i.e.\ set bits correspond to closed and accessed cycle gates, while unset bits correspond to open and skipped cycle gates,
\item $\dlroot(\sigma) = \bitselector$, i.e.\ the strategy is rooted in $\bitselector$,
\item $\sigma(\starteven) = \mu b$ and $\sigma(\bitselector) = \nu b$, i.e.\ one root is connected to the new set bit and the other one is still connected to the former least set bit,
\item $\sigma(\upperselector_i) = \min (\{j > i \mid b_j = 1\} \cup \{n+1\})$, i.e.\ the selectors move to the former next set bit,
\item\label{phase 3 condition open not optimal} $\sigma(\cyclenode_j) \not= \starteven$ for all $j > \mu b$ with $b_j = 0$, i.e.\ every higher open cycle node is not connected to the best-valued root node.
\end{enumerate}

\begin{lemma}\label{lemma: from phase 2 to phase 3}
Let $\sigma$ be a $b$-phase 2 strategy. Then $\sigma'$ is a $b$-phase 3 strategy.
\end{lemma}

\begin{proof}
Again, this can be shown essentially as the previous Lemmas \ref{lemma: from phase 1 to phase 1}
and \ref{lemma: from phase 1 to phase 2}. The main difference is that now
$\cycleaccess_i \prec_\sigma \cycleaccess_{\mu b}$ for all $i \not=\mu b$
which is why $\sigma'(\starteven) = \mu b$, and that by Lemma~\ref{lemma: cycle gate access behaviour}(\ref{lemma: cycle gate access behaviour, item: closed accessed})
we have that the $\mu b$-th gate is $\sigma'$-accessed.
\end{proof}

The cycle gate with the best valuation is now $\mu b$, hence, there are many
improving switches, that eventually lead to cycle gate $\mu b$. First, there are
all nodes of the deceleration lane that have improving switches to $\starteven$.
Second, $\bitselector$ has an improving switch to $\mu b$. Third, lower closed
cycles (all lower cycles should be closed!) have an improving switch to $\mu b$
(opening them again). Fourth, all lower selector nodes have an improving switch
to $\mu b$. By performing all these switches, we enter phase four.

The fourth and last phase, called the \emph{reset} phase, corresponds to a strategy
$\sigma$ that performed the full increment, i.e.\ $b_\sigma = b \oplus 1$. However,
the access states are not reset, i.e.\ $a_\sigma = b[\mu b\mapsto 1]$ and the
deceleration lane is moving to root $\starteven$. 

\noindent More formally, we say that $\sigma$ is a $b$-phase 4 strategy iff all the following conditions hold:
\begin{enumerate}[(1)]
\item\label{phase 4 condition bits} $b \oplus 1 = b_\sigma$ and $b[\mu b\mapsto 1] = a_\sigma$, i.e.\ set bits correspond to closed and accessed cycle gates, while unset bits correspond to open and skipped ($> \mu b$) resp.\ accessed ($< \mu b$) cycles gates,
\item\label{phase 4 deceleration root} $\dlroot(\sigma) = \starteven$, i.e.\ the strategy is rooted in $\starteven$,
\item\label{phase 4 root connection} $\sigma(\starteven) = \sigma(\bitselector) = \mu b$, i.e.\ both roots are connected to the new set bit,
\item\label{phase 4 upper selector} $\sigma(\upperselector_i) = \min (\{j > i \mid (b \oplus 1)_j = 1\} \cup \{n+1\})$, i.e.\ the selectors move to the new next set bit,
\item\label{phase 4 condition index} $\dlindex(\sigma) = 0$, i.e.\ the deceleration lane has reset, and
\item\label{phase 4 condition open not optimal} $\sigma(\cyclenode_j) = \starteven$ for all $j$ with $(b \oplus 1)_j = 0$, i.e.\ every open cycle node is connected to the $\starteven$ root.
\end{enumerate}

\begin{lemma}\label{lemma: from phase 3 to phase 4}
Let $\sigma$ be a $b$-phase 3 strategy. Then $\sigma'$ is a $b$-phase 4 strategy.
\end{lemma}

\proof
Let $\sigma$ be a $b$-phase 3 strategy, $\Xi := \Xi_\sigma$ and $b' := b[\mu b\mapsto 1]$.

We first compute the valuations
for all those nodes directly that do not involve any complicated strategy
decision of player~1. Obviously, $\Xi(\finalcycle) = \emptyset$.
By Lemma~\ref{lemma: simple cycle profitability}(\ref{lemma: simple cycle profitability, item: closed closed})
we know that for all set bits $i$ (i.e.\ $b'_i = 1$) we have the following.
\begin{align*}
\Xi(\cyclecenter_i) = \{\cyclecenter_i\} \cup \Xi(\cycleleaver_i) & &
\Xi(\cyclenode_i) = \{\cyclecenter_i,\cyclenode_i\} \cup \Xi(\cycleleaver_i) & &
\Xi(\cycleaccess_i) = \{\cyclecenter_i,\cycleaccess_i\} \cup \Xi(\cycleleaver_i)
\end{align*}

Using these equations, we are able to compute many other valuations that do not involve any complicated strategy decision of player~1.
Let $U_j = \{\cycleselector_j, \cycleaccess_j, \cyclecenter_j, \cycleleaver_j, \upperselector_j\}$. 
\begin{align*}
\Xi(\upperselector_i) &= \{\upperselector_i\} \cup \bigcup \{U_j \mid j{>}i,b_j{=}1\} &
\Xi(\cycleleaver_i)   &= \{\cycleleaver_i, \upperselector_i\} \cup \bigcup \{U_j \mid j{>}i,b_j{=}1\} \\
\Xi(\cycleselector_i) &= \{\cycleselector_i,\upperselector_i\} \cup \bigcup \{U_j \mid j \geq i,b_j{=}1\} \cup \CF_{i=\mu b}U_i & 
\Xi(\bitselector)     &= \{\bitselector\} \cup \bigcup \{U_j \mid b_j{=}1\} \\
\Xi(\starteven)       &= \{\starteven\} \cup \bigcup \{U_j \mid j{\geq}\mu b,b'_j{=}1\} \setminus \{\cycleselector_{\mu b}\} &
\Xi(\declaneroot)     &= \{\declaneroot, \bitselector\} \cup \bigcup \{U_j \mid b_j=1\} \\
\Xi(\declaneodd_i)    &= \{\declaneodd_i\} \cup \Xi(\bitselector) \cup \CF_{i < \dlindex(\sigma)}(\{\declaneodd_j \mid j < i\} \cup \{\declaneroot\}) &
\Xi(\declaneeven_i)   &= \{\declaneeven_i\} \cup \Xi(\declaneodd_i)
\end{align*}\vspace{-2 pt}

\noindent We have the following orderings on the nodes specified above.
\begin{equation}\tag{a}\label{internal: eq 1}
\bitselector \prec_\sigma \declaneeven_* \prec_\sigma \cycleleaver_{*< \mu b} \prec_\sigma \starteven \prec_\sigma \cycleleaver_{*\geq \mu b}
\end{equation}

Note that the last inequality $\starteven \prec_\sigma \cycleleaver_{i\geq \mu b}$ holds for the following reason:
If $i$ corresponds to a set bit, then the path from $\starteven$ eventually reaches the node $\cycleleaver_i$, but
the highest priority on the way to $\cycleleaver_i$ is $\cycleaccess_i$, which is odd. If $i$ on the other hand
corresponds to an unset bit, then path from $\starteven$ to the sink shares the common postfix with $\cycleleaver_i$,
which starts with the node $\sigma(\upperselector_i)$. Comparing the two differing prefixes shows that the most significant
difference is $\cycleleaver_i$ itself, which is even.

By Lemma~\ref{lemma: simple cycle profitability}(\ref{lemma: simple cycle profitability, item: open closed}), it follows from (\ref{internal: eq 1})
that $\tau_\sigma(\cyclecenter_i) = \cyclenode_i$
for all unset bits $i$ (i.e.\ $b'_i = 0$), hence we are able to compute the
valuations of the remaining nodes.
\begin{align*}
\Xi(\cyclecenter_i) = \{\cyclecenter_i\} \cup \Xi(\cyclenode_i) & &
\Xi(\cycleaccess_i) = \{\cyclecenter_i,\cycleaccess_i\} \cup \Xi(\cyclenode_i)
\end{align*}

It is easy to see that for every $i$ with $(b\oplus 1)_i = 0$ and every $j$ with
$(b\oplus 1)_j = 1$ s.t.\ there is no $i < i' < j$ with $(b\oplus 1)_{i'} = 1$, the following holds:
\begin{equation}\tag{b}\label{internal: eq 2}
\cycleaccess_i \prec_\sigma \cycleaccess_j \qquad
\cycleselector_i \prec_\sigma \cycleselector_j
\end{equation}
Similarly, for $i > j$ with $(b\oplus 1)_i = 1$ and $(b\oplus 1)_j = 1$ we have
\begin{equation}\tag{c}\label{internal: eq 3}
\cycleaccess_i \prec_\sigma \cycleaccess_j \qquad
\cycleselector_i \prec_\sigma \cycleselector_j
\end{equation}


\noindent We are now ready to prove that $\sigma'$ is of the desired form.
\begin{enumerate}[(1)]
\item[(\ref{phase 4 condition bits})] By Lemma~\ref{lemma: simple cycle behaviour}(\ref{lemma: simple cycle behaviour, item: closed closed}) and (\ref{internal: eq 1}) we derive that closed cycles with index $i \geq \mu b$ remain closed.
By Lemma~\ref{lemma: simple cycle behaviour}(\ref{lemma: simple cycle behaviour, item: closed open}) and (\ref{internal: eq 1}) we derive that closed cycles with index $i < \mu b$ open.
By Lemma~\ref{lemma: cycle gate access behaviour}(\ref{lemma: cycle gate access behaviour, item: closed accessed}) we derive that closed cycles remain accessed.
By (\ref{internal: eq 1}) and
Lemma~\ref{lemma: cycle gate access behaviour}(\ref{lemma: cycle gate access behaviour, item: open skipped}) we derive that open cycles remain skipped.

By phase 3 condition 
(\ref{phase 3 condition open not optimal}) and (\ref{internal: eq 1}), it follows
that for every $j$ with $b_j = 0$, there is the improving node $\starteven$
for $\cyclenode_j$. By Lemma~\ref{lemma: simple cycle behaviour}(\ref{lemma: simple cycle behaviour, item: open open}), we conclude that open cycles remain open.

\item[(\ref{phase 4 deceleration root})] By (\ref{internal: eq 1}) and
Lemma~\ref{lemma: deceleration lane behaviour}(\ref{lemma: deceleration lane behaviour, item: changing}).

\item[(\ref{phase 4 root connection})] By (\ref{internal: eq 2}) and (\ref{internal: eq 3}).

\item[(\ref{phase 4 upper selector})] By (\ref{internal: eq 2}) and (\ref{internal: eq 3}).

\item[(\ref{phase 4 condition index})] By Lemma~\ref{lemma: deceleration lane behaviour}(\ref{lemma: deceleration lane behaviour, item: changing}).

\item[(\ref{phase 4 condition open not optimal})] By (\ref{internal: eq 1}) and Lemma~\ref{lemma: simple cycle behaviour}(\ref{lemma: simple cycle behaviour, item: open open}).\qed
\end{enumerate}\vspace{2 pt}

By switching the lane back to
the initial configuration and the access states to match the simple cycles states,
we end up in phase 1 again that corresponds to the incremented global counter state.

\begin{lemma}\label{lemma: from phase 4 to phase 1}
Let $\sigma$ be a $b$-phase 4 strategy and $b \oplus 1 \not= \boldone_n$.
Then $\sigma'$ is a $b \oplus 1$-phase 1 strategy with $\dlindex(\sigma') = 1$.
\end{lemma}

\proof

Let $\sigma$ be a $b$-phase 4 strategy, $\Xi := \Xi_\sigma$ and
$b' = b \oplus 1$.

We first compute the valuations
for all those nodes directly that do not involve any complicated strategy
decision of player~1. Obviously, $\Xi(\finalcycle) = \emptyset$.
By Lemma~\ref{lemma: simple cycle profitability}(\ref{lemma: simple cycle profitability, item: closed closed})
we know that for all set bits $i$ (i.e.\ $b'_i = 1$) we have the following.
\begin{align*}
\Xi(\cyclecenter_i) = \{\cyclecenter_i\} \cup \Xi(\cycleleaver_i) & &
\Xi(\cyclenode_i) = \{\cyclecenter_i,\cyclenode_i\} \cup \Xi(\cycleleaver_i) & &
\Xi(\cycleaccess_i) = \{\cyclecenter_i,\cycleaccess_i\} \cup \Xi(\cycleleaver_i)
\end{align*}
Using these equations, we are able to compute many other valuations that do not involve any complicated strategy decision of player~1.
Let $U_j = \{\cycleselector_j, \cycleaccess_j, \cyclecenter_j, \cycleleaver_j, \upperselector_j\}$. The following holds:
\begin{align*}
\Xi(\upperselector_i) &= \{\upperselector_i\} \cup \bigcup \{U_j \mid j{>}i,b'_j{=}1\} &
\Xi(\cycleleaver_i)   &= \{\cycleleaver_i, \upperselector_i\} \cup \bigcup \{U_j \mid j{>}i,b'_j{=}1\} \\
\Xi(\bitselector)     &= \{\bitselector\} \cup \bigcup \{U_j \mid b'_j=1\} &
\Xi(\starteven)       &= \{\starteven\} \cup (\bigcup \{U_j \mid b'_j=1\} \setminus \{\cycleselector_{\mu b}\}) \\
\Xi(\declaneroot)     &= \{\declaneroot, \starteven\} \cup \bigcup \{U_j \mid b'_j=1\} &
\Xi(\declaneodd_i)    &= \{\declaneodd_i\} \cup \Xi(\starteven) \\
\Xi(\declaneeven_i)   &= \{\declaneeven_i,\declaneodd_i\} \cup \Xi(\starteven) & &
\end{align*}\vspace{-2 pt}

\noindent Additionally for all $i \geq \mu b$, we have:
\begin{displaymath}
\Xi(\cycleselector_i) = \{\cycleselector_i,\upperselector_i\} \cup \bigcup \{U_j \mid j \geq i,b'_j=1\}
\end{displaymath}
It is easy to see that we have the following orderings on the nodes specified above.
\begin{equation}\tag{a}\label{internal: eq 1}
\starteven \prec_\sigma \declaneeven_* \prec_\sigma \bitselector \prec_\sigma \cycleleaver_*
\end{equation}
By Lemma~\ref{lemma: simple cycle profitability}(\ref{lemma: simple cycle profitability, item: open open}), it follows from (\ref{internal: eq 1})
that $\tau_\sigma(\cyclecenter_i) = \cyclenode_i$
for all unset bits $i$ (i.e.\ $b'_i = 0$), hence we are able to compute the
valuations of the remaining nodes.
\begin{align*}
\Xi(\cyclenode_i) = \{\cyclenode_i\} \cup \Xi(\starteven) & &
\Xi(\cyclecenter_i) = \{\cyclecenter_i,\cyclenode_i\} \cup \Xi(\starteven) & &
\Xi(\cycleaccess_i) = \{\cycleaccess_i,\cyclecenter_i,\cyclenode_i\} \cup \Xi(\starteven)
\end{align*}
Additionally for all $i < \mu b$, we have:
\begin{displaymath}
\Xi(\cycleselector_i) = \{\cycleselector_i,\cycleaccess_i,\cyclecenter_i,\cyclenode_i\} \cup \Xi(\starteven)
\end{displaymath}
This completes the valuation of $\Xi$ for all nodes.

It is easy to see that for every $i$ with $b'_i = 0$ and every $j$ with
$b'_j = 1$ s.t.\ there is no $i < i' < j$ with $b'_{i'} = 1$, the following holds:
\begin{equation}\tag{b}\label{internal: eq 2}
\cycleaccess_i \prec_\sigma \cycleaccess_j \qquad
\cycleselector_i \prec_\sigma \cycleselector_j
\end{equation}
Similarly, for $i > j$ with $b'_i = 1$ and $b'_j = 1$ we have
\begin{equation}\tag{c}\label{internal: eq 3}
\cycleaccess_i \prec_\sigma \cycleaccess_j \qquad
\cycleselector_i \prec_\sigma \cycleselector_j
\end{equation}\vspace{-2 pt}

\noindent We are now ready to prove that $\sigma'$ is of the desired form.
\begin{enumerate}[(1)]
\item[(\ref{phase 1 condition bits})] By Lemma~\ref{lemma: simple cycle behaviour}(\ref{lemma: simple cycle behaviour, item: closed closed}) and (\ref{internal: eq 1}) we derive that closed cycles remain closed.
By Lemma~\ref{lemma: cycle gate access behaviour}(\ref{lemma: cycle gate access behaviour, item: closed accessed}) we derive that closed cycles remain accessed.
By (\ref{internal: eq 1}) and
Lemma~\ref{lemma: cycle gate access behaviour}(\ref{lemma: cycle gate access behaviour, item: open skipped}) we derive that open cycles remain or will be skipped.

By Lemma~\ref{lemma: simple cycle behaviour}(\ref{lemma: simple cycle behaviour, item: open open}) and (\ref{internal: eq 1}), we conclude that open cycles remain open.

\item[(\ref{phase 1 deceleration root})] By (\ref{internal: eq 1}) and
Lemma~\ref{lemma: deceleration lane behaviour}(\ref{lemma: deceleration lane behaviour, item: assembling}).

\item[(\ref{phase 1 root connection})] By (\ref{internal: eq 2}) and (\ref{internal: eq 3}).

\item[(\ref{phase 1 upper selector})] By (\ref{internal: eq 2}) and (\ref{internal: eq 3}).

\item[(\ref{phase 1 condition index})] By Lemma~\ref{lemma: deceleration lane behaviour}(\ref{lemma: deceleration lane behaviour, item: changing})
it follows that $\dlindex(\sigma') = 1$.

\item[(\ref{phase 1 condition open not optimal})] By (\ref{internal: eq 1}) it follows that $\sigma'(\cyclenode_i) = \bitselector$ for every $i$ with $b'_i = 0$.\qed
\end{enumerate}

\subsection{Lower Bound Proof}

Finally, we are ready to prove that our family of games really implements a binary
counter. From Lemmas \ref{lemma: from phase 1 to phase 1},
\ref{lemma: from phase 1 to phase 2}, \ref{lemma: from phase 2 to phase 3},
\ref{lemma: from phase 3 to phase 4} and \ref{lemma: from phase 4 to phase 1},
we immediately derive the following.

\begin{lemma}
Let $\sigma$ be a phase 1 strategy and $b_\sigma \not= \boldone_n$. There is
some $k \geq 4$ s.t.\ $\sigma' = \left(\mathcal{I}^\mathtt{loc}\right)^k(\sigma)$
is a phase 1 strategy and $b_{\sigma'} = b_\sigma \oplus 1$.
\end{lemma}

Particularly, we conclude that strategy improvement with the locally optimizing
policy requires exponentially many iterations on $G_n$.

\begin{theorem}
Let $n > 0$. The Strategy Improvement Algorithm with the $\mathcal{I}^\mathtt{loc}$-policy
requires at least $2^n$ improvement steps on $G_n$ starting with $\iota_{G_n}$.
\end{theorem}

\subsection{Remarks}

One could conjecture that 1-sink games form a ``degenerate'' class of parity games as
they are always won by player~1. Remember that the problem of solving parity games
is to determine the complete winning sets for both players. Given a strategy $\sigma$
of player~0, we know by Theorem~\ref{theorem: discrete strategy improvement winning theorem}
that both winning sets can be directly
inferred if $\sigma$ is the optimal strategy. But it is also possible to derive some
information about player~0's winning set given a non-optimal strategy. More precisely,
$W_0 \supseteq \{v \mid \Xi_\sigma(v) = (w, \_, \_) \textrm{ and } w \in V_\oplus\}$.

In other words: Is there a family of games on which the strategy improvement algorithm
requires exponentially many iterations to find a player~0 strategy that wins at least one
node in the game?

The answer to this question is positive. Simply take our lower bound games $G_n$ and
remove the edge from $\cyclecenter_n$ to $\cycleleaver_n$. Remember that
the first time player~1 wants to use this edge by best response is when the binary
counter is about to flip bit $n$, i.e.\ after it processed $2^{n-1}$ many counting
steps. Eventually, the player~0 strategy is updated s.t.\ $\sigma(\cyclenode_n) = \cyclecenter_n$,
forcing player~1 by best response to move to $\cycleleaver_n$. Removing this
edge leaves player~1 no choice but to stay in the cycle which is dominated by player~0.

\section{Improving the Lower Bound Construction} \label{section: improvements}
We briefly address two improvements of our construction. First, we explain how
to reduce the number of edges s.t.\ the overall size of the games is linear in
$n$. Second, we describe how to obtain a lower bound construction with binary
edge outdegree.

\subsection{Linear Number of Edges}
Consider the lower bound construction again. It consists of a deceleration lane,
cycle gates, two roots and connectives between these structures. All three kinds
of structures only have linearly many edges when considered on their own. The
quadratic number of edges is solely due to the $\cyclenode_*$-nodes of the
simple cycles of the cycle gates that are connected to the deceleration lane
and due to the $\upperselector_*$-nodes of the cycle gates that are connected
to all higher cycle gates.

We focus on the edges connecting the $\cyclenode_*$-nodes with the deceleration
lane first. Their purpose is twofold: lower cycle gates have less edges to
the deceleration lane (so they close first), and as long as an open cycle gate
should be prevented from closing, there must be a directly accessible lane input
node in every iteration with a better valuation than the currently chosen lane
input node.

Instead of connecting $\cyclenode_i$ to all $\declaneeven_j$ with $j < 2i + 1$
nodes, it would suffice to connect $\cyclenode_i$ to two intermediate nodes,
say $y_i$ and $z_i$, that are controlled by player~0 with negligible priorities. We
connect $z_i$ to all $\declaneeven_j$ with even $j < 2i + 1$ and $y_i$ to all
$\declaneeven_j$ with odd $j < 2i + 1$. By this construction, we shift the
``busy updating''-part alternately to $y_i$ and $z_i$, and $\cyclenode_i$ remains
updating as well by switching from $y_i$ to $z_i$ and vice versa in every iteration.

Next, we observe that the edges connecting $y_i$ (resp.\ $z_i$) to the lane are
a proper subset of the edges connecting $y_{i+1}$ (resp.\ $z_{i+1}$) to the lane,
and hence we adapt our construction in the following way. Instead of connecting
$y_{i+1}$ (and similarly $z_{i+1}$) to all $\declaneeven_j$ with even $j < 2i + 3$,
we simply connect $y_{i+1}$ to $\declaneeven_{2i+1}$ and to $y_i$. In order to
ensure proper resetting of the two intermediate lanes constituted by $y_*$ and $z_*$
in concordance with the original deceleration, we need to connect every additional
node to $\declaneroot$. See Figure~\ref{figure: deterministic strategy improvement alternating ladders} for the
construction (note that by introducing new nodes with ``negligible priorities'',
we simply shift all other priorities in the game).

\begin{figure}[!h]
\begin{center}
\fbox{
\scalebox{0.55}{
\includegraphics{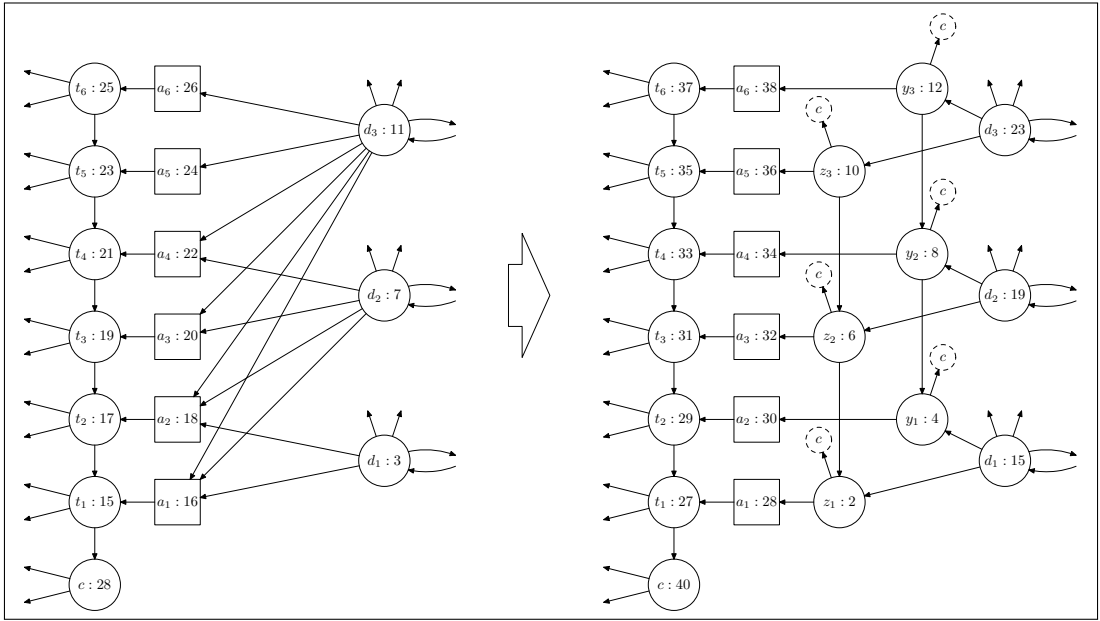}
}
}
\end{center}
\caption{Intermediate Layer}
\label{figure: deterministic strategy improvement alternating ladders}
\end{figure}

Second, we consider the edges connecting lower cycle gates with higher cycle
gates. As the set of edges connecting $\upperselector_{i+1}$ with higher $\cycleselector_j$
is a proper subset of $\upperselector_i$, we can apply a similar construction
by attaching an additional lane to cycle gate connections that subsumes shared
edges.

\subsection{Binary Outdegree}

Every parity game can be linear-time reduced to an equivalent (in the sense that
winning sets and strategies can be easily related to winning sets and strategies
in the original game) parity game with an edge outdegree bounded by two. See
Figure~\ref{figure: deterministic strategy improvement outdegree two transformation} for an example of such a
transformation.

\begin{figure}[!h]
\begin{center}
\fbox{
\includegraphics{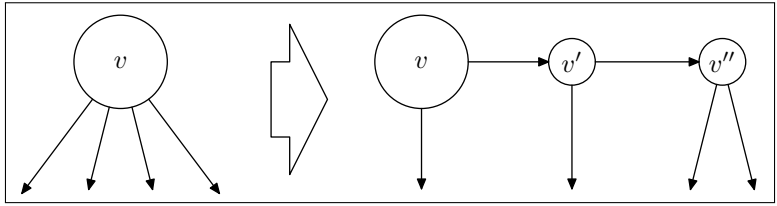}
}
\end{center}
\caption{Binary Outdegree Transformation}
\label{figure: deterministic strategy improvement outdegree two transformation}
\end{figure}

However, not \emph{every} such transformation that can be applied to our
construction (for clarity of presentation, we start with our original construction
again) yields games on which strategy iteration still requires an
exponential number of iterations. We discuss the necessary transformations for
every player~0 controlled node in the following, although we omit the exact
priorities of additional helper nodes. It suffices to assign arbitrary even
priorities to the additional nodes that lie below the priorities of all other 
nodes of the original game (except for the 1-sink).

First, we consider the two root nodes $\starteven$ and $\bitselector$, that
are connected to the 1-sink $\finalcycle$ and to $\cycleaccess_1$,$\ldots$,
$\cycleaccess_n$ resp.\ $\cycleselector_1$,$\ldots$,$\cycleselector_n$.
As $\bitselector$ copies the decision (see the transition from the \emph{access}
to the \emph{reset} phase) of $\starteven$, it suffices to describe how the outdegree-two
transformation is to be applied to $\starteven$.
We introduce $n$ additional
helper nodes $\starteven_1'$,$\ldots$,$\starteven_n'$, replace the outgoing
edges of $\starteven$ by $\finalcycle$ and $\starteven_n'$, connect
$\starteven_{i+1}'$ with $\cycleaccess_{i+1}$ and $\starteven_i'$, and finally
$\starteven_1'$ simply with $\cycleaccess_1$.

It is still possible to show that $\starteven$ reaches the best valued
$\cycleaccess_i$ after one iteration. Assume that $\starteven$ currently reaches
some cycle gate $i$ via the ladder that is given by the helper nodes. Let $j$ be
the next best-valued cycle gate that just has been set. If $j > i$, it follows
that $\starteven$ currently reaches $\starteven_j'$ that moves to $\starteven_{j-1}'$,
but updates within one iteration to $\cycleaccess_j$. If $j < i$, it must be the
case that $j = 1$ ($i$ is the least bit which was set; $j$ is the least bit which
was unset). Moreover, $\starteven$ currently reaches $\starteven_i'$ that moves to
$\cycleaccess_i$. All lower $\starteven_{k+1}'$ with $k + 1 < i$ move to
$\starteven_k'$ since lower unset cycle gates are more profitable than higher
unset cycle gates (unset cycle gates eventually reach one of the roots via the
unprofitable $\cycleaccess_*$ nodes). Hence, $\starteven_i'$ updates within one
iteration to $\starteven_{i-1}'$.

Second, there are the output nodes of cycle gates $\upperselector_1$,$\ldots$,
$\upperselector_n$. We apply a very similar ladder-style construction here.
For every $\upperselector_i$, we introduce $n - i$ additional helper nodes
$\upperselector_{i,j}'$ with $i < j \leq n$, replace the outgoing edges of
$\upperselector_i$ by $\finalcycle$ and $\upperselector_{i,i+1}'$, connect
$\upperselector_{i,j}'$ with $\cycleselector_j$ and $\upperselector_{i,j+1}'$
(if $j < n$). The argument why this construction suffices runs similarly as
for the root nodes.

Third, there are the nodes $\declaneodd_1$,$\ldots$,$\declaneodd_{2n}$ of the
deceleration lane that are connected to three nodes. Again, we introduce an
additional helper node $\declaneodd_i'$ for every $\declaneodd_i$, and replace
the two edges to $\bitselector$ and $\declaneodd_{i-1}$ resp.\ $\declaneroot$
by an edge to $\declaneodd_i'$ that is connected to $\bitselector$ and
$\declaneodd_{i-1}$ resp.\ $\declaneroot$ instead. It is not hard to see that
this slightly modified deceleration lane still provides the same functionality.

Finally, there are the player~0 controlled nodes $\cyclenode_1$,$\ldots$,$\cyclenode_n$
of the simple cycles of the cycle gates. Essentially, two transformations
are possible here. Both replace  $\cyclenode_i$ by as many helper
nodes $\cyclenode_{i,x}'$ as there are edges from $\cyclenode_i$ to any other
node $x$ but $\cyclecenter_i$. Then, every $\cyclenode_{i,x}'$ is connected to
the target node $x$.

The first possible transformation connects every $\cyclenode_{i,x}'$ with
$\cyclecenter_i$ and vice versa, yielding a multicycle with $\cyclecenter_i$ as
the center of each cycle. The second possible transformation connects
$\cyclecenter_i$ with the first $\cyclenode_{i,x_1}'$, $\cyclenode_{i,x_1}'$ with
$\cyclenode_{i,x_2}'$ etc.\ and the last $\cyclenode_{i,x_l}'$ again with
$\cyclecenter_i$, yielding one large cycle. Both replacements behave exactly as
the original simple cycle.

The transformation described here results in a quadratic number of nodes since we
started with a game with a quadratic number of edges.
We note, however, that a similar transformation can be applied to the version of the game
with linearly many edges, resulting in a game with binary outdegree of linear size.

\section{Lower Bound for the Globally Optimizing Policy} \label{section: globally pol}
The lower bound construction for the globally optimizing policy again is a family of
1-sink parity games that implement a binary counter by a combination of a
(modified) deceleration lane and a chain of (modified) cycle gates

This section is organized as follows. First, we discuss the modifications of the
deceleration lane and the cycle gates and why they are required to obtain a
lower bound for the globally optimizing policy.
Then, we present the full construction along with some remarks to the correctness.

The main difference between the locally optimizing policy and the globally
optimizing policy is that the latter takes cross-effects of improving switches
into account. It is aware of the impact of any combination of profitable edges,
in contrast to the locally optimizing policy that only sees the local valuations,
but not the effects.

One primary example that separates both policies are the
simple cycles of the previous sections: the locally optimizing policy sees
that closing a cycle is an improvement, but not that the actual profitability
of closing a cycle is much higher than updating to another node of the
deceleration lane.

The globally optimizing policy, on the other hand, is well aware of the
profitability of closing the cycle in one step. In some sense, the policy has
the ability of a \emph{one-step lookahead}. However, our lower bound for the
globally optimizing policy is not so different from the original construction --
the trick is to hide very profitable choices by structures that cannot be solved
by a single strategy iteration. In other words, we simply need to replace the
gadgets that can be solved with a one-step lookahead by slightly more complicated
variations that cannot be solved within one iteration \emph{and} that maintain
this property for as long as it is necessary.

\subsection{Modified Deceleration Lane}

The modified deceleration lane looks almost the same as the original deceleration lane.
It has again several, say $m$, input nodes $\declaneeven_1,\ldots,\declaneeven_m$
along with some special input node $\declaneroot$. We have two output \emph{roots},
$\bitselector$ and $\starteven$, this time with a slightly different connotation.
We call $\bitselector$ the \emph{default root} and $\starteven$ the \emph{reset root}.

More formally, a modified deceleration lane consists of $m$ (in our case, $m$
will be $6 \cdot n - 2$) internal nodes $\declaneodd_1$, $\ldots$, $\declaneodd_m$,
$m$ input nodes $\declaneeven_1$,
$\ldots$, $\declaneeven_m$, one additional input node $\declaneroot$, the
default root output node $\bitselector$ and the reset root output node $\starteven$.

All priorities of the modified deceleration lane are based on some odd priority $p$. We
assume that all root nodes have a priority greater than $p + 2m + 1$. The
structural difference between the modified deceleration lane and the original
one is that the lane base $\declaneroot$ only has one outgoing edge leading to
the default root $\bitselector$.
See Figure~\ref{figure: globally optimal policy deceleration lane} for a
deceleration lane with $m = 5$ and $p = 27$.
The players, priorities and edges are described in Table~\ref{table: globally optimal policy deceleration lane}.

\begin{table}[!h]
\begin{center}
\tabcolsep5pt
\renewcommand\arraystretch{1.2}
\begin{tabular}{l|l|c|l}
  Node & Player & Priority & Successors \\
  \hline
  $\declaneodd_1$ & $0$ & $p$ & $\{\starteven,\ \bitselector,\ \declaneroot\}$ \\
  $\declaneodd_{i>1}$ & $0$ & $p + 2i - 2$ & $\{\starteven,\ \bitselector,\ \declaneodd_{i - 1}\}$ \\
  $\declaneroot$ & $1$ & $p + 2m + 1$ & $\{\bitselector\}$ \\
  $\declaneeven_i$ & $1$ & $p + 2i - 1$ & $\{\declaneodd_i\}$ \\
  $\starteven$ & ? & $> p + 2m + 1$ & ? \\
  $\bitselector$ & ? & $> p + 2m + 1$ & ?
\end{tabular}
\end{center}
\caption{Description of the Modified Deceleration Lane}
\label{table: globally optimal policy deceleration lane}
\end{table}

\begin{figure}[!h]
\begin{center}
\fbox{
\includegraphics{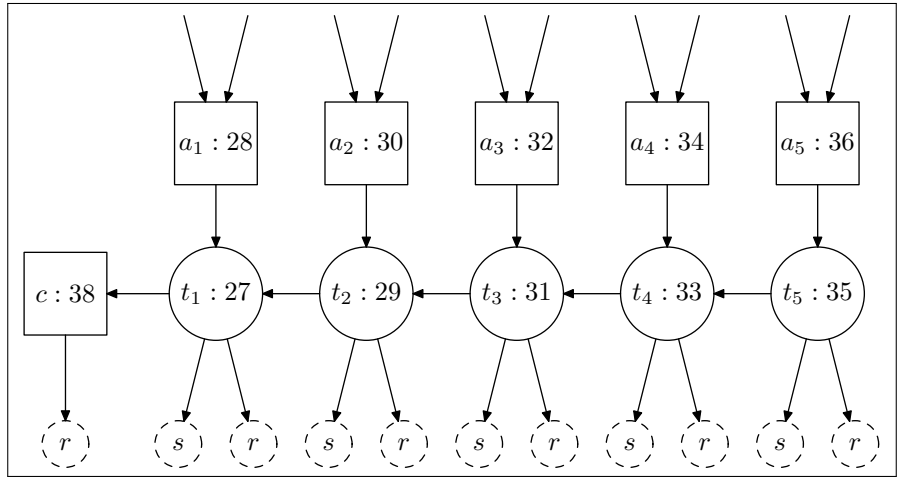}
}
\end{center}
\caption{A Modified Deceleration Lane (with $m = 5$ and $p = 27$)}
\label{figure: globally optimal policy deceleration lane}
\end{figure}

The intuition behind the two roots is the same as before. The default root $\bitselector$
serves as an entry point to the cycle gate structure and the reset root $\starteven$ is
only used for a short time to reset the whole deceleration lane structure.

We describe the state of a modified deceleration
lane again by a tuple specifying which root has been chosen and by how many 
$\declaneodd_i$ nodes are already moving down to $\declaneroot$. Formally, we say
that \emph{$\sigma$ is in deceleration state $(x,j)$} (where $x \in \{\starteven, \bitselector\}$ and
$0 < j \leq m+1$ a natural number) iff
\begin{enumerate}[(1)]
\item $\sigma(\declaneodd_1) = \declaneroot$ if $j > 1$,
\item $\sigma(\declaneodd_i) = \declaneodd_{i-1}$ for all $1 < i < j$, and
\item $\sigma(\declaneodd_i) = x$ for all $j \leq i$.
\end{enumerate}

The modified deceleration lane treats the two roots differently. If the currently
best-valued root is the reset root, it is the optimal choice for all $\declaneodd_*$-
nodes to directly move to the reset root. In other words, no matter what state
the deceleration lane is currently in, if the reset root provides the best valuation,
it requires exactly one improvement step to reach the optimal setting.

If the currently best-valued root is the default root, however, it is profitable
to reach the root via the lane base $\declaneroot$. The globally optimizing
policy behaves in this case just like the locally optimizing policy, because
the deceleration lane has exactly one improving switch at a time which is also
globally profitable.

The following lemma formalizes the intuitive description of the deceleration
lane's behaviour: a change in the ordering of the root valuations leads to a 
reset of the deceleration lane, otherwise the lane continues to align its edges
to eventually reach the best-valued root node via $\declaneroot$.

It is notable that resetting the lane by an external event (i.e.\ by giving
$\starteven$ a better valuation than $\bitselector$) is a bit more difficult
than in the case of the locally optimizing policy. Let $\sigma$ be a strategy and
$\sigma' = \mathcal{I}^\mathtt{glo}(\sigma)$. Assume that the current state
of the deceleration lane is $(\bitselector,i)$ and now we have that $\starteven$
has a better valuation than $\bitselector$, i.e.\ $\starteven \succ_\sigma \bitselector$.
Assume further -- which for instance applies to our original lower bound construction --
that the next strategy $\sigma'$ assigns a better valuation to $\bitselector$ again,
i.e.\ $\bitselector \succ_{\sigma'} \starteven$. Therefore, it would \emph{not}
be the globally optimal choice to reset the deceleration lane to $\starteven$,
but instead just to keep the original root $\bitselector$.

In other words, the globally optimizing policy refrains from resetting the lane
if the resetting event persists for
only one iteration. The solution to fool the policy, however, is not too
difficult: we just alter our construction in such a way that the resetting root
will have a better valuation than the default root for \emph{two}
iterations.

\begin{lemma}
Let $\sigma$ be a strategy that is in deceleration state $(x,i)$. Let $\bar x$
denote the other root. Let $\sigma' = \mathcal{I}^\mathtt{glo}(\sigma)$. Then
\begin{enumerate}[\em(1)]
\item $\bitselector \succ_\sigma \starteven$, $x = \bitselector$ implies that $\sigma'$ is in state $(\bitselector,\min(m,i)+1)$.
\item $\bar x \succ_\sigma x$ and $\bar x \succ_{\sigma'} x$ implies that $\sigma'$ is in state $(\bar x,1)$.
\end{enumerate}
\end{lemma}

The purpose of the modified deceleration lane is exactly the same as before:
we absorb the update activity of cyclic structures that represent the counting
bits of the lower bound construction.

\subsection{Stubborn Cycles}

With the locally optimizing policy, we
employed simple cycles and hid the fact that the improving edge leading into
the simple cycle results in a much better valuation than updating to the next
best-valued node of the deceleration lane.

However, simple cycles do not suffice to fool the globally optimizing policy. If
it is possible to close the cycle within one iteration, the policy sees
that closing the cycle is much more profitable than updating to the deceleration lane.

The solution to this problem is to replace the simple cycle structure by a cycle
consisting of more than one player~0 node s.t.\ it is impossible to close the
cycle within one iteration. More precisely, we use a structure consisting again
of one player~1 node $\cyclecenter$ and three player~0 nodes $\cyclenodex$,
$\cyclenodey$ and $\cyclenodez$, called \emph{stubborn cycle}. We connect all
four nodes with each other in such a way that they form a cycle, and connect all
player~0 nodes with the deceleration lane. See Figure~\ref{figure: globally optimal policy stubborn cycle}
for an example of such a situation.

\begin{figure}[!h]
\begin{center}
\fbox{
\includegraphics{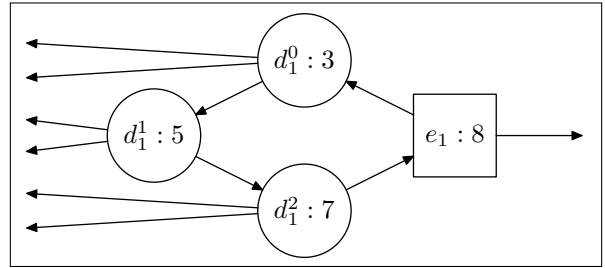}
}
\end{center}
\caption{A Stubborn Cycle}
\label{figure: globally optimal policy stubborn cycle}
\end{figure}

More precisely, we connect the player~0 nodes in a \emph{round robin} manner to the
deceleration lane, for instance $\cyclenodex$ to $\declaneeven_3,\declaneeven_6,\ldots$,
$\cyclenodey$ to $\declaneeven_2,\declaneeven_5,\ldots$, and $\cyclenodez$ to
$\declaneeven_1,\declaneeven_4,\ldots$. We assume that it is more profitable for
player~1 to move into the cyclic structure as long as it is not closed.

Now let $\sigma$ be a strategy s.t.\
$\sigma$ is in state $(\bitselector, 6)$ and $\sigma(\cyclenodex) = \declaneeven_3$,
$\sigma(\cyclenodey) = \cyclenodez$ and $\sigma(\cyclenodez) = \declaneeven_4$.
There are exactly two improving switches here: $\cyclenodey$ to $\declaneeven_5$
(which is the best-valued deceleration node) and $\cyclenodex$ to $\cyclenodey$
(because $\cyclenodey$ currently reaches $\declaneeven_4$ via $\cyclenodez$ which
has a better valuation than $\declaneeven_3$). In fact, the combination of both
switches is the optimal choice.

A close observation reveals that the improved strategy has essentially the same
structure as the original strategy $\sigma$: two nodes leave the stubborn cycle
to the deceleration lane and one node moves into the stubborn cycle. By this
construction, we can ensure that cycles are not closed within one iteration.
In other words, the global policy makes no progress towards closing the cycle
(it switches one edge towards the cycle, and one edge away from the cycle, leaving
it in the exact same position).

\subsection{Modified Cycle Gate}
We again use a slightly modified version of the cycle gates as a pass-through
structure that is either very profitable or quite unprofitable. Essentially,
we apply two modifications. First, we replace the simple
cycle by a stubborn cycle, for the reasons outlined in the previous
subsection. Second, we put an additional player~0 controlled internal node
$\cycleselectorx_i$ between the input node $\cycleselector_i$ and the internal
node $\cycleaccess_i$. It will delay the update of $\cycleselector_i$ to move
to the stubborn cycle after closing the cycle by one iteration. By this, we
ensure that the modified deceleration lane will have enough time to reset itself.

Formally, a modified cycle gate consists of three internal nodes $\cyclecenter_i$,
$\cycleleaver_i$ and $\cycleselectorx_i$,
two input nodes $\cycleaccess_i$ and $\cycleselector_i$, and four output nodes
$\cyclenodex_i$, $\cyclenodey_i$, $\cyclenodez_i$ and $\upperselector_i$.
The output node $\cyclenodex_i$ (resp.\ $\cyclenodey_i$ and $\cyclenodez_i$) will be
connected to a set of other nodes $D_i^1$ (resp.\ $D_i^2$ and $D_i^3$) in the
game graph, and $\upperselector_i$ to some set $K_i$.

All priorities of the cycle gate are based on two odd priorities $p_i$ and
$p_i'$.
See Figure~\ref{figure: globally optimal cycle gate} for a cycle gate of index~$1$ with $p_1' = 3$
and $p_1 = 65$.
The players, priorities and edges are described in Table~\ref{table: globally optimal cycle gate}.

\begin{table}[!h]
\begin{center}
\tabcolsep5pt
\renewcommand\arraystretch{1.2}
\begin{tabular}{l|l|c|l}
  Node & Player & Priority & Successors \\
  \hline
  $\cyclenodex_i$ & $0$ & $p_i'$ & $\{\cyclenodey_i\} \cup D_i^1$ \\
  $\cyclenodey_i$ & $0$ & $p_i'+2$ & $\{\cyclenodez_i\} \cup D_i^2$ \\
  $\cyclenodez_i$ & $0$ & $p_i'+4$ & $\{\cyclecenter_i\} \cup D_i^3$ \\
  $\cyclecenter_i$ & $1$ & $p_i'+5$ & $\{\cyclenodex_i,\ \cycleleaver_i\}$ \\
  $\cycleselectorx_i$ & $0$ & $p_i'+6$ & $\{\cycleaccess_i,\ \upperselector_i\}$ \\
  $\cycleselector_i$ & $0$ & $p_i'+7$ & $\{\cycleselectorx_i,\ \upperselector_i\}$ \\
  $\upperselector_i$ & $0$ & $p_i$ & $K_i$ \\
  $\cycleaccess_i$ & $1$ & $p_i+2$ & $\{\cyclecenter_i\}$ \\
  $\cycleleaver_i$ & $1$ & $p_i+3$ & $\{\upperselector_i\}$
\end{tabular}
\end{center}
\caption{Description of the Modified Cycle Gate}
\label{table: globally optimal cycle gate}
\end{table}

\begin{figure}[!h]
\begin{center}
\fbox{
\includegraphics{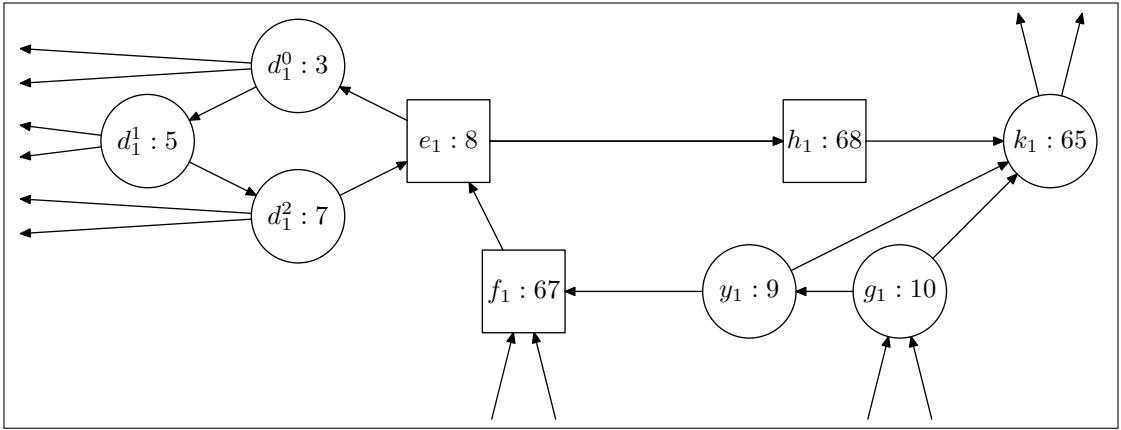}
}
\end{center}
\caption{A Modified Cycle Gate (index $1$ with $p'_1 = 3$ and $p_1 = 65$)}
\label{figure: globally optimal cycle gate}
\end{figure}

From an abstract point of view, we describe the state of a modified cycle gate
again by a 
pair $(\beta_i(\sigma),\alpha_i(\sigma)) \in \{0,1,2,3\}\times\{0,1,2\}$. The first component
describes the state of the
stubborn cycle, counting the number of edges pointing into the cycle,
and the second component gives the state of the two access control nodes.
Formally, we have the following.
\begin{displaymath}
\beta_i(\sigma) = |\{\cyclenode_i^j \mid \sigma(\cyclenode_i^j) \not\in D_i^j\}|
\qquad\qquad
\alpha_i(\sigma) = \begin{cases}
2 & \text{if } \sigma(\cycleselector_i) = \cycleselectorx_i \\
0 & \text{if } \sigma(\cycleselector_i) = \sigma(\cycleselectorx_i) = \upperselector_i \\
1 & \text{otherwise}
\end{cases}
\end{displaymath}
The behaviour is formalized in terms of modified cycle gate states as follows. 
Intuitively, it functions as the original cycle gates: if the cycle is
$\sigma$-closed and remains closed, it is profitable to
go through the cycle gate. If the cycle opens by some external event
and remains open, it is more profitable to directly move
to the output node instead.
\begin{lemma}
Let $\sigma$ be a strategy and $\sigma' = \mathcal{I}^\mathtt{glo}(\sigma)$.
\begin{enumerate}[\em(1)]
\item If $\beta_i(\sigma) = \beta_i(\sigma') = 3$, we have $\alpha_i(\sigma') = \min(\alpha_i(\sigma)+1,2)$
      (``closed gates will be successively accessed'').
\item If $\beta_i(\sigma) < 3$, $\beta_i(\sigma') < 3$ and $\sigma(\upperselector_i) \succ_{\sigma'} \cycleaccess_i$,
      we have $\alpha_i(\sigma') = 0$ (``open gates with unprofitable exit nodes will be skipped'').
\end{enumerate}
\end{lemma}

\noindent We use modified cycle gates again to represent the bit states of a binary
counter: unset bits will correspond to modified cycle gates with the state $(1,0)$,
set bits to the state $(3,2)$. Setting and resetting bits
therefore traverses more than one phase, more precisely, from $(1,0)$
over $(2,0)$, $(3,0)$ and $(3,1)$ to $(3,2)$, and from the
latter again over $(1,2)$ to $(1,0)$.\vspace{-3 pt}

\subsection{Modified Construction}
In this subsection, we provide the complete construction of the lower
bound family for the globally optimizing policy. It again consists of
a 1-sink $\finalcycle$, a modified deceleration lane of length $6n-3$
that is connected to the two roots $\starteven$ and $\bitselector$,
and $n$ modified cycle gates. The stubborn cycles of the cycle gates
are connected to the $\bitselector$ root, the lane base $\declaneroot$
and to the deceleration lane. The modified cycle gates are connected
to each other in the same manner as in the original lower bound
structure for the locally optimizing policy.

The way the stubborn cycles are connected to the deceleration lane is
more involved as in the previous lower bound construction. Remember
that for all open stubborn cycles, we need to maintain the setting in
which two edges point to the deceleration lane while the other points
into the cycle. We achieve this task by assigning the three nodes of
the respective stubborn cycle to the input nodes of the deceleration
lane in a round-robin fashion.

We now give the formal construction.  The games are denoted by
$H_n=(V_n,V_{n,0},V_{n,1},E_n,\Omega_n)$.  The sets of nodes are
\begin{flalign*}
V_n := &\{\finalcycle, \starteven, \declaneroot, \bitselector\} \cup \{\declaneeven_i, \declaneodd_i \mid 0 < i \leq 6n-2\} \cup\\& \{\cyclenodex_i, \cyclenodey_i, \cyclenodez_i, \cyclecenter_i, \cycleaccess_i, \cycleleaver_i, \cycleselector_i, \cycleselectorx_i, \upperselector_i \mid 0 < i \leq n\}
\end{flalign*}

The players, priorities and edges are described in Table~\ref{table: globally optimizing lower bound game}.
The game $H_3$ is depicted in Figure~\ref{figure: globally optimizing lower bound game}.
However, the edges connecting the cycle gates with the deceleration
lane are not included in the figure.

\begin{table}[!h]
\begin{center}
\tabcolsep5pt
\renewcommand\arraystretch{1.2}
\begin{tabular}{l|l|c|l}
  Node & Player & Priority & Successors \\
  \hline
  $\declaneodd_1$ & $0$ & $8n + 3$ & $\{\starteven,\ \bitselector,\ \declaneroot\}$ \\
  $\declaneodd_{i>1}$ & $0$ & $8n + 2i + 1$ & $\{\starteven,\ \bitselector,\ \declaneodd_{i - 1}\}$ \\
  $\declaneeven_i$ & $1$ & $8n + 2i + 2$ & $\{\declaneodd_i\}$ \\
  $\declaneroot$ & $1$ & $20n$ & $\{\bitselector\}$ \\
  \hline
  $\cyclenodex_i$ & $0$ & $8i + 1$ & $\{\starteven,\ \declaneroot, \cyclenodey_i\} \cup \{\declaneeven_{3j+3} \mid j \leq 2i-2\}$ \\
  $\cyclenodey_i$ & $0$ & $8i + 3$ & $\{\cyclenodez_i\} \cup \{\declaneeven_{3j+2} \mid j \leq 2i-2\}$ \\
  $\cyclenodez_i$ & $0$ & $8i + 5$ & $\{\cyclecenter_i\} \cup \{\declaneeven_{3j+1} \mid j \leq 2i-1\}$ \\
  $\cyclecenter_i$ & $1$ & $8i + 6$ & $\{\cyclenodex_i,\ \cycleleaver_i\}$ \\
  $\cycleselectorx_i$ & $0$ & $8i + 7$ & $\{\cycleaccess_i,\ \upperselector_i\}$ \\
  $\cycleselector_i$ & $0$ & $8i + 8$ & $\{\cycleselectorx_i,\ \upperselector_i\}$ \\
  $\upperselector_i$ & $0$ & $20n + 4i + 3$ & $\{\finalcycle\} \cup \{\cycleselector_{j} \mid i < j \leq n\}$ \\
  $\cycleaccess_i$ & $1$ & $20n + 4i + 5$ & $\{\cyclecenter_i\}$ \\
  $\cycleleaver_i$ & $1$ & $20n + 4i + 6$ & $\{\upperselector_i\}$ \\
  \hline
  $\starteven$ & $0$ & $20n + 2$ & $\{\cycleaccess_j \mid j \leq n\} \cup \{\finalcycle\}$ \\
  $\bitselector$ & $0$ & $20n + 4$ & $\{\cycleselector_j \mid j \leq n\} \cup \{\finalcycle\}$ \\
  $\finalcycle$ & $1$ & $1$ & $\{\finalcycle\}$
\end{tabular}
\end{center}
\caption{Lower Bound Construction for the Globally Optimizing Policy}
\label{table: globally optimizing lower bound game}
\end{table}

\begin{figure}[!h]
\begin{center}
\rotatebox{90}{
\scalebox{0.87}{
\includegraphics{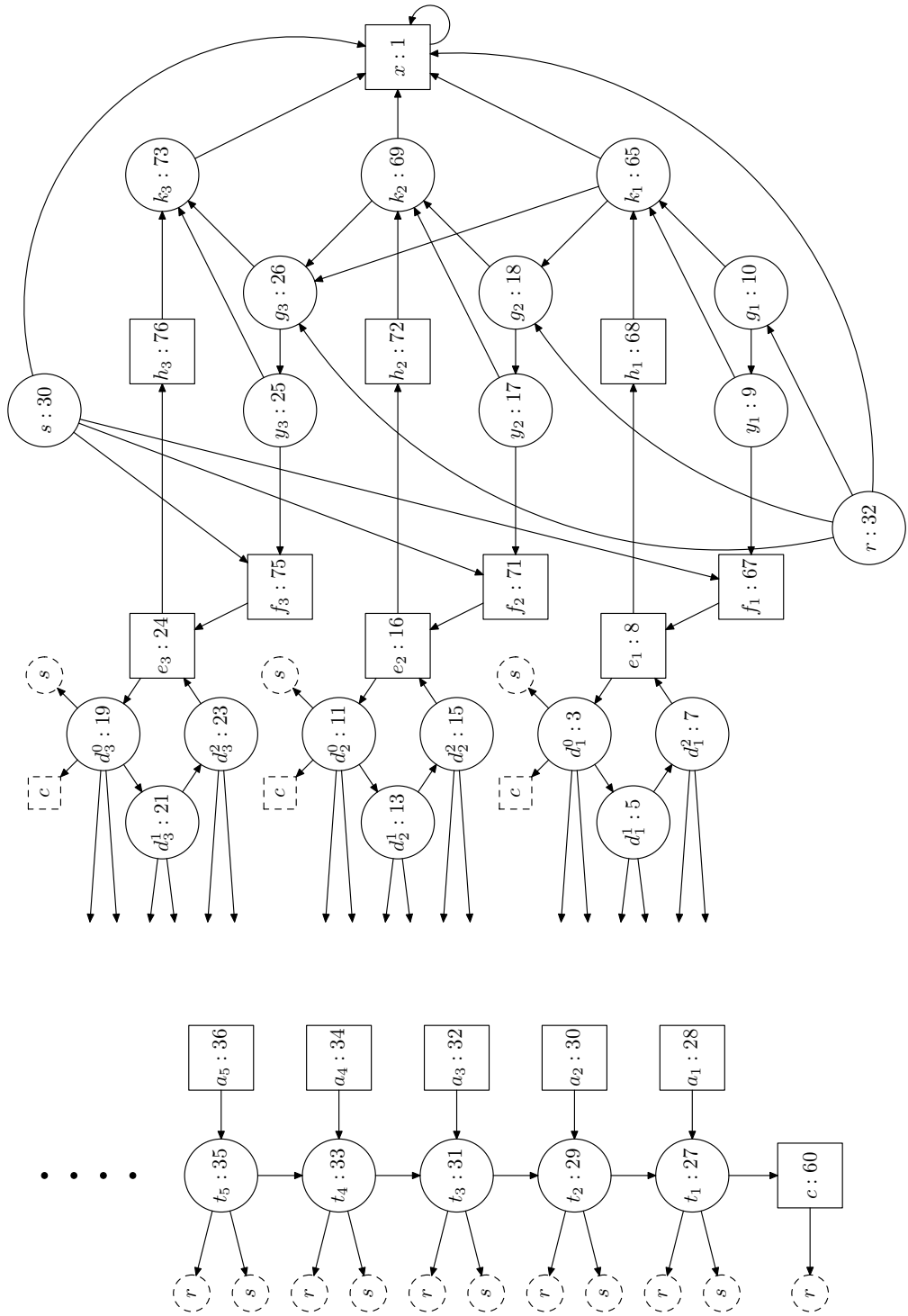}
}
}
\end{center}
\caption{Globally Optimizing Lower Bound Game $H_3$}
\label{figure: globally optimizing lower bound game}
\end{figure}

\begin{fact} \label{size of mod game}
The game $H_n$ has $21 \cdot n$ nodes, $3.5 \cdot n^2 + 40.5 \cdot n - 4$ edges
and $24 \cdot n + 6$ as highest priority. In particular, $|H_n| = \BigO{n^2}$.
\end{fact}

As an initial strategy we select the following strategy $\iota_{H_n}$. Again, it
corresponds to a global counter setting in which no bit has been set.
\begin{flalign*} \qquad
 & \iota_{H_n}(\declaneodd_1) = \declaneroot & \quad & \iota_{H_n}(\declaneodd_{1 < i \leq 3}) = \declaneodd_{i-1} & \quad & \iota_{H_n}(\declaneodd_{i > 3}) = \bitselector & \\
 & \iota_{H_n}(\declaneroot) = \bitselector & \quad & \iota_{H_n}(\cyclenodex_i) = \cyclenodey_i & \quad & \iota_{H_n}(\cyclenodey_i) = \declaneeven_2 & \\
 & \iota_{H_n}(\cyclenodez_i) = \declaneeven_1 & \quad & \iota_{H_n}(\cycleselector_i) = \upperselector_i & \quad & \iota_{H_n}(\cycleselectorx_i) = \upperselector_i & \\
 & \iota_{H_n}(\upperselector_i) = \finalcycle & \quad & \iota_{H_n}(\starteven) = \finalcycle & \quad & \iota_{H_n}(\bitselector) = \finalcycle &
\end{flalign*}

\noindent It is easy to see that the $H_n$ family again is a family of 1-sink games.

\begin{lemma} \label{mod game is 1 sink}
Let $n > 0$.
\begin{enumerate}[\em(1)]
\item The game $H_n$ is completely won by player 1.
\item $\finalcycle$ is the 1-sink of $H_n$ and the cycle component of $\Xi_{\iota_{H_n}}(w)$ equals $\finalcycle$ for all $w$.
\end{enumerate}
\end{lemma}

Again, we note that it is possible to refine the family $H_n$ in such a
way that it only comprises a linear number of edges and only outdegree two.

\subsection{Remarks}

The way to prove the construction corrects runs almost exactly the
same as for the locally optimizing policy. Every global counting step
is separated into some counting iterations of the deceleration lane
with busy updating of the open stubborn cycles of the cycle gates
until the least significant open cycle closes. Then, resetting of the
lane, reopening of lower cycles and alignment of connecting edges is
carried out.

\begin{theorem}
Let $n > 0$. The Strategy Improvement Algorithm with the $\mathcal{I}^\mathtt{glo}$-policy
requires at least $2^n$ improvement steps on $H_n$ starting with $\iota_{H_n}$.
\end{theorem}

Our publicly available \pgsolver Collection \cite{pgsolver} of parity
game solvers contains implementations of strategy iteration, and particularly
parameterizations with the locally and globally optimizing policy.
Additionally, the platform features a number of game generators, including all
the games and extensions that are presented here.
Benchmarking both strategy iteration variants with our lower bound constructions
results in exponential run-time behavior as can be seen in Figure~\ref{figure: strategy improvement benchmark}.

\begin{figure}[t]
\scalebox{0.6}{\includegraphics[trim = 1cm 2cm 10cm 5cm, clip]{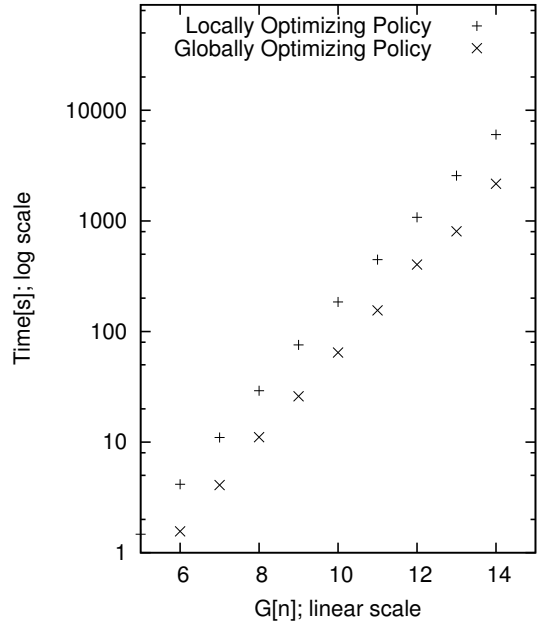}}
\caption{Empirical Evaluation}
\label{figure: strategy improvement benchmark}
\end{figure}

\section{Mean Payoff, Discounted Payoff and Simple Stochastic Games} \label{section: othergames}
We now show that the standard reductions \cite{puri/phd,zwickpaterson/1996} from parity games
to mean payoff, discounted payoff as well as simple stochastic games can be used to derive 
worst-case families for all the other game classes.

A \emph{mean payoff game} is a tuple $G = (V, V_0, V_1, E, r)$ where $V$, $V_0$, $V_1$ and $E$
are as in the definition of parity games and $r: V \rightarrow \RR$ is the so-called
\emph{reward function}. A \emph{discounted payoff game} is a tuple $G = (V, V_0, V_1, E, r, \beta)$ where
$(V, V_0, V_1, E, r)$ is a mean payoff game and $0 < \beta < 1$ is the so-called \emph{discount factor}.
Whenever we do not want to distinguish between a discounted and a mean payoff game, we simply write
payoff game.

Strategies and plays are defined exactly the same as in the definition of parity games. Given a play $\pi$,
the \emph{payoff of the play} $R_G(\pi)$ is defined as follows. For a mean payoff game $G = (V, V_0, V_1, E, r)$,
it is
\begin{displaymath}
R_G(\pi) := \liminf_{n \rightarrow \infty} \frac{1}{n} \sum_{k=0}^n r(\pi_k)
\end{displaymath}
and in the case of a discounted payoff game $G = (V, V_0, V_1, E, r, \beta)$, it is
\begin{displaymath}
R_G(\pi) := \sum_{k=0}^\infty \beta^k \cdot r(\pi_k)
\end{displaymath}

Let $G$ be a payoff game. For a given node $v$, a player 0 strategy $\sigma$ and a player 1 strategy $\varrho$,
let $\pi_{v,\sigma,\varrho}$ denote the unique play that starts in $v$ and conforms to $\sigma$ and $\varrho$.
We say that a node $v$ \emph{has a value} iff $\sup_\sigma \inf_\varrho R_G(\pi_{v,\sigma,\varrho})$ and 
$\inf_\varrho \sup_\sigma R_G(\pi_{v,\sigma,\varrho})$ exist, and
\begin{displaymath}
\sup_\sigma \inf_\varrho R_G(\pi_{v,\sigma,\varrho}) = \inf_\varrho \sup_\sigma R_G(\pi_{v,\sigma,\varrho})
\end{displaymath}
Whenever a node $v$ has a value, we write $\vartheta_G(v) := \sup_\sigma \inf_\varrho R_G(\pi_{v,\sigma,\varrho})$ to
refer to it. If every node has a value, we say that a player~0 strategy $\sigma$ is \emph{optimal} iff
$\inf_\varrho R_G(\pi_{v,\sigma,\varrho}) \geq \inf_\varrho R_G(\pi_{v,\sigma',\varrho})$ for every node $v$ and
every player~0 strategy $\sigma'$ and similarly for player~1.

\begin{theorem}[\cite{EhMy1979}]
Let $G$ be a mean payoff game. Every node $v$ has a value and there are \emph{optimal} positional
strategies $\sigma$ and $\varrho$ s.t.\ $\vartheta_G(v) = R_G(\pi_{v,\sigma,\varrho})$ for every $v$.
\end{theorem}
Note that given two optimal positional strategies, it is fairly easy to compute the associated values.

Parity Games can be easily polynomial-time reduced to mean payoff games s.t.\ optimal strategies correspond to winning
strategies and the values of the nodes directly induce corresponding winning sets in the original
parity game. Given a parity game $G = (V, V_0, V_1, E, \Omega)$, the \emph{$G$-induced mean payoff game}
$\inducedmpg{G} = (V, V_0, V_1, E, r_\Omega)$ operates on the same graph and defines the reward function $r_\Omega$ as
follows.
\begin{displaymath}
r_\Omega: v \mapsto (-|V|)^{\Omega(v)}
\end{displaymath}

\begin{theorem}[\cite{puri/phd}] \label{theorem: pg mpg reduction}
Let $G$ be a parity game and let $\sigma$ and $\varrho$ be optimal
positional strategies w.r.t.\ $\inducedmpg{G}$. Then the following holds.
\begin{enumerate}[\em(1)]
\item $W_0 = \{v \in V \mid \vartheta_{\inducedmpg{G}}(v) \geq 0\}$ is the $G$-winning set of player 0
\item $W_1 = \{v \in V \mid \vartheta_{\inducedmpg{G}}(v) < 0\}$ is the $G$-winning set of player 1
\item $\sigma$ is a $G$-winning strategy for player 0 on $W_0$
\item $\varrho$ is a $G$-winning strategy for player 1 on $W_1$
\end{enumerate}
\end{theorem}

Mean payoff games can be reduced to discounted payoff games by specifying a discount factor that is sufficiently
close to 1. Given a mean payoff game $G = (V, V_0, V_1, E, r)$, the \emph{$G$-induced discounted payoff game}
$\induceddpg{G} = (V, V_0, V_1, E, r, \beta_G)$ operates on the same graph with the same reward function and defines
the discount factor $\beta_G$ as follows. We assume that $r: V \rightarrow \ZZ$.
\begin{displaymath}
\beta_G := 1 - \frac{1}{4 \cdot |V|^3 \cdot \max \{|r(v)| \mid v \in V \}}
\end{displaymath}
Every parity game $G$ obviously also induces a discounted payoff game via an intermediate mean payoff game.

Let $v$ be a node in a mean payoff game $G$, let $\vartheta(v)$ be the value of $v$ in $G$, and let
$\vartheta_\beta(v)$ be the value of $v$ in $\induceddpg{G}$.
Zwick and Paterson \cite{zwickpaterson/1996} show that the value $\vartheta(v)$ can be essentially bounded by
$\vartheta_\beta(v)$, i.e.\ $|\vartheta_\beta(v) - \vartheta(v)| \leq \frac{1-\beta}{2|V|^2(1-\beta_G)}$.
By choosing $\beta \geq \beta_G$, it follows that $\vartheta(v)$ can be obtained from $\vartheta_\beta(v)$ by
rounding to the nearest rational with a denominator less than $|V|$. It follows
that optimal strategies in an induced discouned payoff game coincide with optimal
strategies in the original mean payoff game.

\begin{theorem}[\cite{zwickpaterson/1996}] \label{theorem: mpg dpg reduction}
Let $G$ be a mean payoff game and let $\sigma$ and $\varrho$ be optimal
positional strategies w.r.t.\ $\induceddpg{G}$. Then $\sigma$ and $\varrho$ are also optimal positional strategies w.r.t.\ $G$.
\end{theorem}

Let $G$ be a discounted payoff game. Every player 0 strategy $\sigma$ induces an optimal (not necessarily unique) counterstrategy
$\varrho_\sigma$ s.t.\ $R_G(\pi_{v,\sigma,\varrho}) \leq R_G(\pi_{v,\sigma,\varrho'})$ for all other player 1 strategies
$\varrho'$ and all nodes $v$. Note that $\varrho_\sigma$ can be computed by solving an $\LP$-problem as described
in Algorithm~\ref{algorithm: optimal counter strategy in dpg}.

\begin{algorithm}[!h]
Maximize $\sum_{v \in V} \varphi(v)$ w.r.t.\ $\qquad \qquad \qquad \qquad \qquad \qquad \qquad \qquad \qquad \qquad \qquad \qquad \qquad \qquad$
\begin{align*}
\varphi(v) &= r(v) + \beta \cdot \varphi(\sigma(v)) &\textrm{ for all } v \in V_0\\
\varphi(v) &\leq r(v) + \beta \cdot \varphi(u) &\textrm{ for all } v \in V_1 \textrm{ and } u \in vE
\end{align*}
\caption{Computation of Optimal Counter Strategy in a DPG}
\label{algorithm: optimal counter strategy in dpg}
\end{algorithm}

The value assignment $\varphi$ can be computed in strongly polynomial time by applying the algorithm
of Madani, Thorup and Zwick \cite{MadaniTZ10} for instance. Given $\varphi$, an optimal counterstrategy $\varrho_\sigma$ can be
easily induced.

We say that a strategy $\sigma$ is \emph{improvable} iff there is a node $v \in V_0$ and a node $u \in vE$
s.t.\ $R_G(\pi_{\sigma(v),\sigma,\varrho_\sigma}) < R_G(\pi_{u,\sigma,\varrho_\sigma})$. Again, an \emph{improvement policy}
is a function $\mathcal{I}_G: \strategyset{G}{0} \rightarrow \strategyset{G}{0}$ that satisfies the following two
conditions for every strategy $\sigma$.
\begin{enumerate}[(1)]
\item For every node $v \in V_0$ it holds that $R_G(\pi_{\sigma(v),\sigma,\varrho_\sigma}) \leq R_G(\pi_{\mathcal{I}_G(\sigma)(v),\sigma,\varrho_\sigma})$.
\item If $\sigma$ is improvable then there is a node $v \in V_0$ s.t.\ $R_G(\pi_{\sigma(v),\sigma,\varrho_\sigma}) < R_G(\pi_{\mathcal{I}_G(\sigma)(v),\sigma,\varrho_\sigma})$.
\end{enumerate}

As with parity game strategy improvement, it is the case that improving a strategy following
improvement edges results indeed in an improved strategy.
\begin{theorem}[\cite{puri/phd}]
Let $G$ be a discounted payoff game, $\sigma$ be a player 0 strategy and $\mathcal{I}_G$ be
an improvement policy. Then $R_G(\pi_{v,\sigma,\varrho_\sigma}) \leq R_G(\pi_{v,\mathcal{I}_G(\sigma),\varrho_{\mathcal{I}_G(\sigma)}})$
for every node $v$. If $\sigma$ is not optimal, then $\sigma$ is improvable.
\end{theorem}

Puri's algorithm for solving discounted payoff games -- as well as mean payoff and parity 
games via the standard reductions -- starts with an initial strategy $\iota_G$ and runs
for a given improvement policy $\mathcal{I}_G$ as outlined in Algorithm~\ref{algorithm: puris algorithm for dpg}.
Note that the algorithmic scheme is exactly the same as the discrete version for solving parity games.

\begin{algorithm}[!h]
\begin{algorithmic}[1]
\State $\sigma \gets \iota_G$
\While {$\sigma$ is improvable}
	\State $\sigma \gets \mathcal{I}_G(\sigma)$
\EndWhile
\State \textbf{return} $\sigma$, $\varrho_\sigma$
\end{algorithmic}
\caption{Puri's Algorithm for Solving Discounted Payoff Games}
\label{algorithm: puris algorithm for dpg}
\end{algorithm}

Next, we will show that the strategy iteration for discounted payoff games behaves exactly the same as
the strategy iteration for 1-sink-parity games. 

V\"oge proves in his thesis \cite{voege/phd} the following theorem that relates parity game strategy
iteration to Puri's Algorithm for solving the induced discounted payoff game.

\begin{theorem}[\cite{voege/phd}]\label{theorem: voege discounted reduction}
Let $G$ be a parity game, $H=\induceddpg{\inducedmpg{G}}$
be the induced discounted payoff game and $\sigma$ be a player~0 strategy.
For every two nodes $v$ and $u$ the following holds.
\begin{displaymath}
\Xi_\sigma(v) \prec \Xi_\sigma(u) \qquad \Rightarrow \qquad R_H(\pi_{v,\sigma,\varrho_\sigma}) < R_H(\pi_{u,\sigma,\varrho_\sigma})
\end{displaymath}
\end{theorem}

In other words, every improving switch in the original parity game is also an improving
switch in the induced discounted payoff game. The reason why this holds true is that by
the reduction from parity games to mean payoff games, the priorities are mapped to such
extremely large rewards that the largest reward that occurs on a path dominates all lower
ones, the largest reward on a cycle dominates all other ones and that the cycle itself
dominates all finite paths leading into it. 

Theorem~\ref{theorem: voege discounted reduction} is almost what we need to show that
strategy iteration for discounted payoff games behaves exactly the same on $\induceddpg{\inducedmpg{G_n}}$
as the discrete strategy iteration algorithm on $G_n$. Essentially, we need to show the
conversion which is equivalent to showing
\begin{displaymath}
\Xi_\sigma(v) = \Xi_\sigma(u) \qquad \Rightarrow \qquad R_H(\pi_{v,\sigma,\varrho_\sigma}) = R_H(\pi_{u,\sigma,\varrho_\sigma})
\end{displaymath}

However, this statement is not true for every parity game. The reason why a run of the strategy improvement algorithm
on general parity games may differ from a run on the induced discounted payoff game is that the parity game
strategy iteration does not care about the priority of all nodes on its path to the dominating cycle node
that are less relevant. In case of 1-sink-parity games, the only occurring dominating cycle node has the
least priority in the game, and therefore all priorities occurring in paths influence the valuations. Also,
the strategy iteration on arbitrary parity games does not consider the priorities of all the nodes on a cycle
appearing in a node valuation.

First, we show that optimal player 1 counter strategies in the induced discounted payoff game also
eventually reach the 1-sink.

\begin{lemma} \label{lemma: sink pg to sink dpg}
Let $G$ be a 1-sink-parity game with $v^*$ being the 1-sink, $H=\induceddpg{\inducedmpg{G}}$
be the induced discounted payoff game, and $\sigma$ be a player
0 strategy s.t.\ $\Xi_{\iota_G} \unlhd \Xi_\sigma$. Let $v_0 \not= v^*$ be an arbitrary node.
Then, $\pi_{v_0,\sigma,\varrho_\sigma}$ is of the following form:
\begin{displaymath}
\pi_{v_0,\sigma,\varrho_\sigma} = v_0 v_1 \ldots v_{l-1} (v^*)^\omega
\end{displaymath}
\end{lemma}

\begin{proof}
Consider the games $G' := G|_\sigma$ and $H' := H|_\sigma$ and note that $\tau^G_\sigma = \tau^{G'}_\sigma$
as well as $\varrho^H_\sigma = \varrho^{H'}_\sigma$. Note that $G'$ is won by player 1 following
$\tau^{G'}_\sigma$ since $G$ is 1-sink parity game.

By Theorems \ref{theorem: pg mpg reduction} and \ref{theorem: mpg dpg reduction} it follows that $\varrho^{H'}_\sigma$ must
be also a player 1 winning strategy for the whole game $G'$. Therefore, it follows that every play
$\pi_{v_0,\sigma,\varrho_\sigma}$ eventually ends in cycle with a dominating cycle node $w^*$ of
odd priority, hence $\Omega(w^*) \geq \Omega(v^*)$.

If $\Omega(w^*) > \Omega(v^*)$, it follows that there is a $w^*$-dominated cycle reachable in
$G'$ starting from $v_0$. But since $\Xi_\sigma(v_0) = (v^*,\_,\_)$, this cannot be the case.
Hence $\Omega(w^*) = \Omega(v^*)$, implying that $w^* = v^*$.
\end{proof}

Second, we show that the value ordering between two different paths leading to the 1-sink again
depends solely on the most relevant node in the symmetric difference of the paths.

\begin{lemma} \label{lemma: sink pg path to sink dpg path}
Let $G$ be a 1-sink-parity game with $v^*$ being the 1-sink, $H=\induceddpg{\inducedmpg{G}}$
be the induced discounted payoff game. Let $\pi$ and $\xi$
be two paths of the form $\pi = u_0 u_1 \ldots u_{l-1} (v^*)^\omega$ and
$\xi = w_0 w_1 \ldots w_{k-1} (v^*)^\omega$ and let $U = \{u_0, \ldots, u_{l-1}\}$
and $W = \{w_0, \ldots, w_{k-1}\}$. Then $U \prec W$ implies $R_H(\pi) \prec R_H(\xi)$.
\end{lemma}

\proof
Let $V = \{v_0, \ldots, v_{n-1}\}$ s.t.\ $p_{n-1} > p_{n-2} > \ldots > p_0$ with 
$p_i = \Omega(v_i)$ and $v_0 = v^*$, and let $\beta$ be the discount factor of $H$.
W.l.o.g.\ assume that $n > 2$ since otherwise both paths are necessarily the same.
Let $a: \{1,\ldots,n-1\} \rightarrow \{0,\ldots,n-2,\bot\}$ be a map s.t.\
\begin{displaymath}
a(i) = \begin{cases}
j & \textrm{ if } u_j = v_i \\
\bot & \textrm{ if there is no $j$ s.t.\ } u_j = v_i
\end{cases}
\end{displaymath}
and let $b: \{1,\ldots,n-1\} \rightarrow \{0,\ldots,n-2,\bot\}$ be defined accordingly for $w_j$.
Set $\beta^\bot := 0$.
Note that the following holds.
\begin{displaymath}
R_H(\pi) = \sum_{i=1}^{n-1} \beta^{a(i)} \cdot (-n)^{p_i} + \frac{n \cdot \beta^l}{\beta - 1} \qquad \qquad
R_H(\xi) = \sum_{i=1}^{n-1} \beta^{b(i)} \cdot (-n)^{p_i} + \frac{n \cdot \beta^k}{\beta - 1}
\end{displaymath}
Let $m = \max \{i \mid (a(i) = \bot \textrm{ and } b(j) \not= \bot) \textrm{ or } (a(i) \not= \bot \textrm{ and } b(j) = \bot)\}$ and note that $m$ indeed is well-defined. Set
\begin{displaymath}
\Delta := R_H(\xi) - R_H(\pi) = \Delta_1 + \Delta_2 + \Delta_3 + \Delta_4
\end{displaymath}
where
\begin{flalign*}
\Delta_1 &:= \sum_{i=m+1}^{n-1} (\beta^{b(i)} - \beta^{a(i)}) \cdot (-n)^{p_i} \\
\Delta_2 &:= (\beta^{b(m)} - \beta^{a(m)}) \cdot (-n)^{p_m} \\
\Delta_3 &:= \sum_{i=1}^{m-1} (\beta^{b(i)} - \beta^{a(i)}) \cdot (-n)^{p_i} \\
\Delta_4 &:= \frac{n \cdot (\beta^k - \beta^l)}{\beta - 1}
\end{flalign*}
Regarding $\Delta_1$, let $m < i < n$ and consider that $|\beta^{b(i)} - \beta^{a(i)}| \leq |1 - \beta^{n-2}|$.
The following holds.
\begin{flalign*}
|\beta^{b(i)} - \beta^{a(i)}| \cdot n^{p_i} &\leq |1 - \beta^{n-2}| \cdot n^{p_i} = (\sum_{j=0}^{n-3}\beta^j) \cdot (1-\beta) \cdot n^{p_i} \\
&\leq n \cdot (1-\beta) \cdot n^{p_i} = \frac{n^{p_i+1}}{4\cdot n^3 \cdot n^{p_{n-1}}} \leq n^{-2}
\end{flalign*}
We conclude that $|\Delta_1| \leq \frac{n-1-m}{n^2} \leq 1$.

Regarding $\Delta_2$, note that $b(m) \not= \bot$ implies that $p_m$ is even and $b(m) = \bot$ implies 
that $p_m$ is odd. Let $c = b(m)$ iff $b(m) \not= \bot$ and $c = a(m)$ otherwise. Hence the following holds.
\begin{flalign*}
\Delta_2 &= \beta^c \cdot n^{p_m} \geq \beta^{n-1} \cdot n^{p_m} = (\beta^{n-1} - 1) \cdot n^{p_m} + n^{p_m} \\
&= (\sum_{j=0}^{n-2} \beta^j) \cdot (\beta - 1) \cdot n^{p_m} + n^{p_m} \geq (1-n) \cdot (1-\beta) \cdot n^{p_m} + n^{p_m} \\
&= \frac{1-n}{4\cdot n^{p_{n-1}+3}} \cdot n^{p_m} + n^{p_m} \geq \frac{3}{4} \cdot n^{p_m}
\end{flalign*}
Regarding $\Delta_3$, let $0 < i < m$ and consider that $|\beta^{b(i)} - \beta^{a(i)}| \leq 1$. The following
holds.
\begin{displaymath}
|\Delta_3| \leq \sum_{i=1}^{m-1} |\beta^{b(i)} - \beta^{a(i)}| \cdot n^{p_i} \leq \sum_{i=1}^{m-1} n^{p_i}
\end{displaymath}
Now we need to distinguish on whether $p_m = 2$. If so, note that $m = 1$, $b(m) \not= \bot$ and
$k = l + 1$. Hence, regarding $\Delta_4$, the following holds.
\begin{displaymath}
|\Delta_4| = \frac{n \cdot |\beta^{l+1} - \beta^l|}{|\beta - 1|} = n \cdot \beta \leq n
\end{displaymath}
Therefore we conclude (remember that $n > 2$)
\begin{displaymath}
\Delta \geq \Delta_2 - |\Delta_1| - |\Delta_3| - |\Delta_4| \geq \frac{3}{4} \cdot n^2 - 1 - 0 - n > 0
\end{displaymath}
Otherwise, if $p_m > 2$, it holds that $|\beta^k - \beta^l| \leq |1-\beta^{n-1}| \leq (n-1) \cdot (1-\beta)$
and hence $|\Delta_4| \leq n^2 - n$. Additionally, consider $\Delta_3$ again.
\begin{displaymath}
|\Delta_3| \leq \sum_{i=1}^{m-1} n^{p_i} \leq \sum_{i=2}^{p_m-1} n^i = \sum_{i=0}^{p_m-1} n^i - 1 - n = \frac{n^{p_m}-1}{n-1} - 1 - n
\end{displaymath}
We conclude the following (remember again that $n > 2$).
\begin{flalign*}
\Delta &\geq \Delta_2 - |\Delta_1| - |\Delta_3| - |\Delta_4| \\
&\geq \frac{3}{4} \cdot n^{p_m} - 1 - \frac{n^{p_m}-1}{n-1} + 1 + n - n ^2 + n \\
&= \frac{3}{4} \cdot n^{p_m} - \frac{n^{p_m}-1}{n-1} - (n - 1)^2 + 1 \\
&\geq \frac{3}{4} \cdot n^{p_m} - \frac{n^{p_m}}{2} - (n - 1)^2 + 1 \\
&= \frac{1}{4} \cdot n^{p_m} - (n - 1)^2 + 1 > 0\rlap{\hbox to 174 pt{\hfill\qEd}}
\end{flalign*}\vspace{2 pt}

\noindent Third, we derive that the strategy iteration for discounted payoff games behaves exactly the same as
the strategy iteration for 1-sink-parity games.

\begin{theorem}
Let $G$ be a 1-sink-parity game, $v$ be a node, $H=\induceddpg{\inducedmpg{G}}$
be the induced discounted payoff game and $\sigma$ be a player
0 strategy s.t.\ $\Xi_{\iota_G} \unlhd \Xi_\sigma$. Then $\varrho_\sigma = \tau_\sigma$.
\end{theorem}

\begin{proof}
Assume by contradiction that $\varrho_\sigma \not= \tau_\sigma$. Hence, there is a node
$v$ s.t.\ $\pi_{v,\sigma,\tau_\sigma} \not= \pi_{v,\sigma,\varrho_\sigma}$. Since $G$ is
a 1-sink parity game and $\Xi_{\iota_G} \unlhd \Xi_\sigma$, it follows by Lemma
\ref{lemma: sink pg to sink dpg} that $\pi_{v,\sigma,\varrho_\sigma}$ eventually
reaches the 1-sink. It follows that
$R_H(\pi_{v,\sigma,\varrho_\sigma}) < R_H(\pi_{v,\sigma,\tau_\sigma})$ which is impossible
due to Lemma \ref{lemma: sink pg path to sink dpg path}.
\end{proof}

\begin{corollary}
Let $G$ be a 1-sink-parity game, $H=\induceddpg{\inducedmpg{G}}$
be the induced discounted payoff game and $\sigma$ be a player
0 strategy s.t.\ $\Xi_{\iota_G} \unlhd \Xi_\sigma$. For every two nodes $v$ and $u$ the
following holds.
\begin{displaymath}
\Xi_\sigma(v) \prec \Xi_\sigma(u) \qquad \iff \qquad R_H(\pi_{v,\sigma,\varrho_\sigma}) < R_H(\pi_{u,\sigma,\varrho_\sigma})
\end{displaymath}
\end{corollary}

\begin{corollary}
Puri's algorithm for solving payoff games requires exponentially many iterations
in the worst case when parameterized with the locally or the globally optimal policy.
\end{corollary}

We note that it is possible to define strategy iteration for mean payoff games directly, i.e.\ without applying the reduction to discounted payoff games first. Unfortunately, with mean payoff games, it is not the case that if $\sigma$ is not optimal then there necessarily exists at least one switch that \emph{strictly} improves the reward. There are several way to remedy this situation; most of them are based on a lexicographic ordering again with the first component being the reward and the second component being a description of the nodes leading to the cycle, usually called \emph{potential}. We note without proof that our results translate to this variant of strategy iteration as well.

Finally, we relate our results to simple stochastic games. Particularly, we consider simple stochastic games with
arbitrary outdegree and arbitrary probabilities that halt almost surely. Zwick, Paterson and Condon
show that there is direct correspondence between this version of simple stochastic games and the original
one \cite{zwickpaterson/1996,condon92thecomplexity}.

A \emph{simple stochastic game} is a tuple $G=(V,V_\mathit{min},V_\mathit{max},V_\mathit{avg},0,1,E,p)$
s.t.\ $V_\mathit{min}$, $V_\mathit{max}$, $V_\mathit{avg}$, $\{0\}$ and $\{1\}$ are a partition of $V$,
$(V,E)$ is a directed graph with exactly two sinks $0$ and $1$, and $p: E \cap (V_\mathit{avg} \times V) \rightarrow [0;1]$
is the probability mapping s.t.\ $\sum_{u \in vE} p(v,u) = 1$ for all $v \in V_\mathit{avg}$. 

We say that a simple stochastic game \emph{halts with probability 1} iff every node $v$ in $G|_{\sigma,\tau}$
has a path with non-negligible probabilities to a sink for every pair of strategies $\sigma$ and $\tau$. Every simple stochastic game
can be reduced to an equivalent simple stochastic game that halts with probability 1 in polynomial time
\cite{condon92thecomplexity}. We assume from now on that every given simple stochastic game halts with
priority 1.

Given a simple stochastic game and a play $\pi$, we say that player $\mathit{Max}$ wins $\pi$ iff it
ends in the 1-sink and similarly that player $\mathit{Min}$ wins $\pi$ if it ends in the 0-sink. 
Let $R_G(v,\sigma,\varrho)$
denote the probability that player $\mathit{Max}$ wins starting from $v$ conforming to the $\mathit{Max}$-strategy
$\sigma$ and the $\mathit{Min}$-strategy $\varrho$.

We say that a node $v$ \emph{has a value} iff $\sup_\sigma \inf_\varrho R_G(\pi_{v,\sigma,\varrho})$ and 
$\inf_\varrho \sup_\sigma R_G(\pi_{v,\sigma,\varrho})$ exist, and
\begin{displaymath}
\sup_\sigma \inf_\varrho R_G(\pi_{v,\sigma,\varrho}) = \inf_\varrho \sup_\sigma R_G(\pi_{v,\sigma,\varrho})
\end{displaymath}
Whenever a node $v$ has a value, we write $\vartheta_G(v) := \sup_\sigma \inf_\varrho R_G(\pi_{v,\sigma,\varrho})$ to
refer to it. If every node has a value, we say that a player~0 strategy $\sigma$ is \emph{optimal} iff
$\inf_\varrho R_G(\pi_{v,\sigma,\varrho}) \geq \inf_\varrho R_G(\pi_{v,\sigma',\varrho})$ for every node $v$ and
every player~0 strategy $\sigma'$ and similarly for player~1.

\begin{theorem}[\cite{condon92thecomplexity}]
Let $G$ be a simple stochastic game. Every node $v$ has a value and there are \emph{optimal} positional
strategies $\sigma$ and $\varrho$ s.t.\ $\vartheta_G(v) = R_G(\pi_{v,\sigma,\varrho})$ for every $v$.
\end{theorem}

Again, simple stochastic games can be solved by strategy iteration. Given a player $\mathit{Max}$ strategy $\sigma$, 
an (not necessarily unique) optimal counterstrategy $\varrho_\sigma$ -- i.e. $R_G(v,\sigma,\varrho) \leq R_G(v,\sigma,\varrho')$ for all
other player $\mathit{Min}$ strategies $\varrho'$ and all nodes $v$ -- can be computed by solving an $\LP$-problem as described
in Algorithm~\ref{algorithm: optimal counter strategy in ssg}.

\begin{algorithm}[!h]
Maximize $\sum_{v \in V} \varphi(v)$ w.r.t.\ $\qquad \qquad \qquad \qquad \qquad \qquad \qquad \qquad \qquad \qquad \qquad \qquad \qquad \qquad$
\begin{align*}
\varphi(v) &= \varphi(\sigma(v)) &\textrm{ for all } v \in V_\mathit{max}\\
\varphi(v) &\leq \varphi(u) &\textrm{ for all } v \in V_\mathit{min} \textrm{ and } u \in vE\\
\varphi(v) &= \sum_{u \in vE} p(v,u) \cdot \varphi(u) &\textrm{ for all } v \in V_\mathit{avg}\\
\varphi(1) &= 1 \\
\varphi(0) &= 0
\end{align*}
\caption{Computation of Optimal Counter Strategy in a SSG}
\label{algorithm: optimal counter strategy in ssg}
\end{algorithm}

The value assignment $\varphi$ can be computed in polynomial time by applying Khachiyan's algorithm
\cite{khachiyan79lp} for instance. Given $\varphi$, an optimal counterstrategy $\varrho_\sigma$ can be
efficiently deduced. The strategy iteration that solves the simple stochastic games runs exactly the same
as for discounted payoff games.

Zwick and Paterson \cite{zwickpaterson/1996} describe a simple reduction from discounted payoff games to
simple stochastic games that halt with probability 1. Let $G = (V,V_0,V_1,E,r,\beta)$ be a discounted payoff
game and let $l = \min \{r(v) \mid v \in V\}$, $u = \max \{r(v) \mid v \in V\}$ and
$d = \max(1,u-l)$.

The \emph{$G$-induced simple stochastic game} is the game
$\inducedssg{G}=(V',V_\mathit{min},V_\mathit{max},V_\mathit{avg},0,1,E',p)$ where $V_\mathit{min} = V_1$,
$V_\mathit{max} = V_0$, $V_\mathit{avg} = E$, $V' = V_\mathit{min} \cup V_\mathit{max} \cup V_\mathit{avg} \cup \{0,1\}$ and
\begin{flalign*}
E' &= \{(v,(v,u)),((v,u),u),((v,u),0),((v,u),1) \mid v \in V \textrm{ and } u \in E\} \\
p &: \begin{cases}
((v,u),u) \mapsto \beta \\
((v,u),1) \mapsto (1 - \beta) \cdot \frac{r(v) - l}{d} \\
((v,u),0) \mapsto (1 - \beta) \cdot (1 - \frac{r(v) - l}{d})
\end{cases}
\end{flalign*}
Clearly, the induced simple stochastic game halts with probability 1. As Zwick and Paterson pointed out,
the values of the induced simple stochastic game directly correspond to the values of the original discounted payoff game.

\begin{lemma}[\cite{zwickpaterson/1996}]
Let $G$ be a discounted payoff game and $G'=\inducedssg{G}$. Let $\sigma$ be a player 0
strategy and $\varrho$ be a player 1 strategy. Then $(1-\beta) \cdot R_G(\pi_{v,\sigma,\varrho}) = d \cdot R_{G'}(v,\sigma,\varrho) + l$
for every node $v$ where $l = \min \{r(v) \mid v \in V\}$, $u = \max \{r(v) \mid v \in V\}$ and $d = \max(1,u-l)$.
\end{lemma}

This particularly implies that $R_G(\pi_{v,\sigma,\varrho}) = \frac{d}{1-\beta} \cdot R_{G'}(v,\sigma,\varrho) + l$ with
$\frac{d}{1-\beta} > 0$, i.e.\ the values of the original discounted payoff game correspond to the values of the
induced simple stochastic game by an affine transformation that preserves the ordering.

\begin{corollary}
The standard strategy iteration for simple stochastic games requires exponentially many iterations
in the worst case.
\end{corollary}

\section{Conclusion} \label{section: conclusion}
We have presented a family of games on which the deterministic strategy improvement algorithm
for parity games requires exponentially many iterations. Additionally, we have
shown how to adapt this family to prove an exponential lower bound on Schewe's
policy.

Finally, we have shown that the presented family can be used to transfer the exponential
lower bound to mean payoff, discounted payoff and simple stochastic games by applying the
standard reductions.

Although there are many preprocessing techniques that could be used to simplify the family
of games presented here -- e.g.\ decomposition into strongly connected components, compression
of priorities, direct-solving of simple cycles, see \cite{FriedmannLangePract09} for instance --
they are no solution to the general weakness of strategy iteration on these games, simply due to
the fact that all known preprocessing techniques can be fooled quite easily without really
touching the inner structure of the games.

Parity games are widely believed to be solvable in polynomial time, yet there is no algorithm known
that is performing better than superpolynomially. Jurdzi{\'n}ski and V\"oge presented the strategy iteration technique
for parity games over ten years ago, and this class of solving procedures is generally supposed to be
the best candidate to give rise to an algorithm that solves parity games in polynomial time since then.
Unfortunately, the locally and the globally optimizing technique are not capable of achieving this goal.

We think that the strategy iteration still is a promising candidate for a polynomial time algorithm,
however it may be necessary to alter more of it than just the improvement policy. 


\paragraph{\bf Acknowledgements.} I am very thankful to Martin Lange and Martin Hofmann for their guidance and numerous inspiring discussions on the subject. Also, I would like to thank the anonymous referees for their thorough reports containing many comments that helped to improve the presentation of this paper.

\bibliographystyle{alpha}
\bibliography{./main}

\end{document}